\documentclass[draftcls,12pt,onecolumn]{IEEEtran}		

\usepackage[utf8]{inputenc}
\usepackage[T1]{fontenc}
\usepackage{url}
\usepackage{ifthen}
\usepackage{cite}
\usepackage[cmex10]{amsmath} 

\usepackage{mypackage}
\usepackage{mathrsfs}

\usepackage{afterpage}

\newtheorem{proposition}{Proposition}
\usepackage{mathtools}

\theoremstyle{approximation}

\usepackage{physics}
\usepackage{siunitx}

\newcommand{\sir}{\mathrm{SIR}}
\newcommand{\pc}{\mathrm{P_c}}
\newcommand{\pcd}{\mathrm{P_{c_d}}}
\newcommand{\pcb}{\mathrm{P_{c_b}}}

\usepackage[nomain,acronym,shortcuts]{glossaries}
\makeglossaries
\newcommand*{\acro}[3][]{\newacronym[#1]{#2}{#2}{#3}}

\acro{D2D}{device-to-device}
\acro{SIR}{signal to interference ratio}
\acro{SINR}{signal to interference plus noise ratio}
\acro{PCP}{Poisson cluster process}
\acro{CoMP}{coordinated multi-point}
\acro{BS}{base station} 
\acro{MD-CoMP}{macrodiversity CoMP transmission}
\acro{MAC}{medium-access-control}
\acro{JT-CoMP}{joint transmission CoMP}
\acro{CoMP-JT}{coordinated multipoint joint transmission}
\acro{SBS}{small base station}
\acro{MDSD}{multiple devices to the single device}
\acro{MDS}{maximum distance seperable}
\acro{SCN}{small cell network}
\acro{PPP}{Poisson point process}
\acro{TCP}{Thomas cluster process}
\acro{CSI}{channel state information}
\acro{PDF}{probability distribution function}
\acro{PMF}{probability mass function}
\acro{RV}{random variable}
\acro{i.i.d.}{independently and identically distributed}
\acro{MBMS}{multimedia broadcasting multicasting service}
\acro{EE}{energy efficiency}
\acro{HCP}{hard-core placement}
\acro{CCDF}{complementary cumulative distribution function}
\acro{PC}{probabilistic caching}
\acro{RC}{random uniform caching}
\acro{CPF}{caching popular files} 
\acro{BCD}{block coordinate descent}
\acro{CTMC}{continuous-time Markov chain}
\acro{PGFL}{point generating functional}
\acro{KKT}{Karush-Kuhn-Tucker}

\begin{document}
\title{Optimizing Joint Probabilistic Caching and Communication for Clustered {D2D} Networks}
 
\author{Ramy Amer,~\IEEEmembership{Student~Member,~IEEE,} Hesham~Elsawy,~\IEEEmembership{Senior~Member,~IEEE,} M.~Majid~Butt,~\IEEEmembership{Senior~Member,~IEEE,} Eduard~A.~Jorswieck,~\IEEEmembership{Senior~Member,~IEEE,}
Mehdi~Bennis,~\IEEEmembership{Senior~Member,~IEEE,}~and~Nicola~Marchetti,~\IEEEmembership{Senior~Member,~IEEE}
 
 \thanks{The material in this paper will be presented in part at Globecom 2018 \cite{amer2018minimizing}.}
\thanks{Ramy Amer and Nicola~Marchetti are with CONNECT Centre for Future Networks, Trinity College Dublin, Ireland. Email:\{ramyr, nicola.marchetti\}@tcd.ie.}
\thanks{Hesham ElSawy is with King Fahd University of Petroleum and Minerals (KFUPM), Saudi Arabia. Email: elsawyhesham@gmail.com.}
\thanks{M. Majid Butt is with Nokia Bell Labs, France, and CONNECT Centre for Future Networks, Trinity College Dublin, Ireland. Email: Majid.Butt@@tcd.ie.}
\thanks{Eduard A. Jorswieck is with Department of Electrical Engineering and Information Technology, TU Dresden, Germany. Email: eduard.jorswieck@tudresden.de.}
\thanks{Mehdi Bennis is with the Centre for Wireless Communications, University of Oulu, Finland, and the Department of Computer Engineering, Kyung Hee University, South Korea. Email: bennis@ee.oulu.fi.}
\thanks{This publication has emanated from research conducted with the financial support of Science Foundation Ireland (SFI) and is co-funded under the European Regional Development Fund under Grant Number 13/RC/2077.}
}

\maketitle			
\begin{abstract} 
Caching at mobile devices and leveraging \ac{D2D} communication are two promising approaches to support massive content delivery over wireless networks. The analysis of such \ac{D2D} caching networks based on a physical interference model is usually carried out by assuming that devices are uniformly distributed. However, this approach does not fully consider and characterize the fact that devices are usually grouped into clusters. Motivated by this fact, this paper presents a comprehensive performance analysis and joint communication and caching optimization for a clustered \ac{D2D} network. Devices are distributed according to a \ac{TCP}  and are assumed to have a surplus memory which is exploited to proactively cache files  from a known library, following a random probabilistic caching scheme. Devices can retrieve the requested files from their caches, from neighboring devices in their proximity (cluster), or from the base station as a last resort. Three key performance metrics are optimized in this paper, namely, the offloading gain, energy consumption, and latency. Firstly, we maximize the offloading probability of the proposed network by jointly optimizing channel access and caching probability. Secondly, we formulate and solve the energy minimization problem for the proposed model and obtain the optimal probabilistic caching for the minimum energy consumption. Finally, we jointly optimize the caching scheme as well as bandwidth allocation between \ac{D2D} and base station-to-Device transmission to minimize the weighted average delay per file request. Employing the \ac{BCD} optimization technique, we propose an efficient iterative algorithm for solving the delay minimization problem. A closed-form solution for the bandwidth allocation sub-problem is also provided. Simulation results show significant improvement in the network performance reaching up to $10\%$, $17\%$, and $300\%$ for the offloading gain, energy consumption, and average delay, respectively compared to the Zipf's caching baseline.
\end{abstract}
\begin{IEEEkeywords}
\ac{D2D} communication, probabilistic caching, offloading gain, energy consumption, delay analysis, stochastic geometry, queuing theory.
\end{IEEEkeywords}
\section{Introduction}
Caching at mobile devices significantly improves system performance by facilitating \ac{D2D} communications, which enhances the spectrum efficiency and alleviate the heavy burden on backhaul links \cite{Femtocaching}. Modeling the cache-enabled heterogeneous networks, including \ac{SBS} and mobile devices, follows two main directions in the literature. The first line of work focuses on the fundamental throughput scaling results by assuming a simple protocol channel model \cite{Femtocaching,golrezaei2014base,8412262,amer2017delay}, known as the protocol model, where two devices can communicate if they are within a certain distance. The second line of work, defined as the physical interference model, considers a more realistic model for the underlying physical layer \cite{andreev2015analyzing,cache_schedule}. In the following, we review some of the works relevant to the second line, focusing mainly on the \ac{EE}  and delay analysis of wireless caching networks. 

The physical interference model is based on the fundamental \ac{SIR} metric, and therefore, is applicable to any wireless communication system. Modeling devices' locations as a \ac{PPP}   is widely employed in the literature, especially, in the wireless caching area \cite{andreev2015analyzing,cache_schedule,energy_efficiency,ee_BS,hajri2018energy}. However, a realistic model for \ac{D2D} caching networks requires that a given device typically has multiple proximate devices, where any of them can potentially act as a serving device. This deployment is known as clustered devices' deployment, which can be characterized by cluster processes \cite{haenggi2012stochastic}. Unlike the popular \ac{PPP} approach, the authors in \cite{clustered_twc,clustered_tcom,8070464} developed a stochastic geometry based model to characterize the performance of content placement in the clustered \ac{D2D} network. In \cite{clustered_twc}, the authors discuss two strategies of content placement in a \ac{PCP}   deployment. First, when each device randomly chooses its serving device from its local cluster, and secondly, when each device connects to its $k$-th closest transmitting device from its local cluster. The authors characterize the optimal number of \ac{D2D} transmitters that must be simultaneously activated in each cluster to maximize the area spectral efficiency.  The performance of cluster-centric content placement is characterized in \cite{clustered_tcom}, where the content of interest in each cluster is cached closer to the cluster center, such that the collective performance of all the devices in each cluster is optimized. Inspired by the Matern hard-core point process, which captures pairwise interactions between nodes, the authors in \cite{8070464}  devised a novel spatially correlated caching strategy called \ac{HCP} such that the \ac{D2D} devices caching the same content are never closer to each other than the exclusion radius. 


Energy efficiency in wireless caching networks is widely studied in the literature \cite{energy_efficiency,ee_BS,hajri2018energy}.
For example, an optimal caching problem is formulated in \cite{energy_efficiency} to minimize the energy consumption of a wireless network. The authors consider a cooperative wireless caching network where relay nodes cooperate with the devices to cache the most popular files in order to minimize energy consumption. In \cite{ee_BS}, the authors investigate how caching at BSs can improve \ac{EE} of wireless access networks. The condition when \ac{EE} can benefit from caching is characterized, and the optimal cache capacity that maximizes the network \ac{EE} is found. It is shown that \ac{EE} benefit from caching depends on content popularity, backhaul capacity, and interference level. 
The authors in \cite{hajri2018energy} exploit the spatial repartitions of devices and the correlation in their content popularity profiles to improve the achievable EE. The \ac{EE} optimization problem is decoupled into two related subproblems, the first one addresses the issue of content popularity modeling, and the second subproblem investigates the impact of exploiting the spatial repartitions of devices. It is shown that the small base station allocation algorithm improves the energy efficiency and hit probability. However, the problem of \ac{EE} for \ac{D2D} based caching is not yet addressed in the literature.

Recently, the joint optimization of delay and energy in wireless caching is conducted, see, for instance \cite{wu2018energy,huang2018energy,jiang2018energy}. The authors in \cite{wu2018energy} jointly optimize the delay and energy in a cache-enabled dense small cell network. The authors formulate the energy-delay optimization problem as a mixed integer programming problem, where file placement, device association to the small cells, and power control are jointly considered. To model the energy consumption and end-to-end file delivery-delay tradeoff, a utility function linearly combining these two metrics is used as an objective function of the optimization problem.  An efficient algorithm is proposed to approach the optimal association and power solution, which could achieve the optimal tradeoff between energy consumption and end-to-end file delivery delay. In \cite{huang2018energy}, the authors showed that with caching, the energy consumption can be reduced by extending transmission time. However, it may incur wasted energy if the device never needs the cached content. Based on the random content request delay, the authors study the maximization of \ac{EE} subject to a hard delay constraint in an additive white Gaussian noise channel. It is shown that the \ac{EE} of a system with caching can be significantly improved with increasing content request probability and target transmission rate compared with the traditional on-demand scheme, in which the \ac{BS} transmits content file only after it is requested by the user. However, the problem of energy consumption and joint communication and caching for clustered \ac{D2D} networks is not yet addressed in the literature.

In this paper, we conduct a comprehensive performance analysis and optimization of the joint communication and caching for a clustered \ac{D2D} network, where the devices have unused memory to cache some files, following a random probabilistic caching scheme. Our network model effectively characterizes the stochastic nature of channel fading and clustered geographic locations of devices. Furthermore, this paper emphasizes on the need for considering the traffic dynamics and rate of requests when studying the delay incurred to deliver requests to devices. To the best of our knowledge, our work is the first in the literature that conducts a comprehensive spatial analysis of a doubly \ac{PCP}  (also called doubly \ac{PPP}  \cite{haenggi2012stochastic}) with the devices adopting a slotted-ALOHA random access technique to access a shared channel. The key advantage of adopting the slotted-ALOHA access protocol is that it is a simple yet fundamental medium access control (MAC) protocol, wherein no central controller exists to schedule the users' transmissions. 
We also incorporate the spatio-temporal analysis in wireless caching networks by combining tools from stochastic geometry and queuing theory in order to analyze and minimize the average delay (see, for instance, \cite{zhong2015stability,zhong2017heterogeneous,7917340}).			
The main contributions of this paper are summarized below.				

\begin{itemize}
\item We consider a Thomas cluster process (TCP)   where the devices are spatially distributed as groups in clusters. The clusters' centers are drawn from a parent PPP, and the clusters' members are normally distributed around the centers, forming a Gaussian PPP. This organization of the parent and offspring PPPs forms the so-called doubly PPP.
\item We conduct the coverage probability analysis where the devices adopt a slotted-ALOHA random access technique. We then jointly optimize the  access probability and caching probability to maximize the cluster offloading gain. We obtain the optimal channel access probability, and then a closed-form solution of the optimal caching sub-problem is provided. The energy consumption problem is then formulated and shown to be convex and the optimal caching probability is also formulated. 
\item By combining tools from stochastic geometry as well as queuing theory, 
we minimize the per request weighted average delay by jointly optimizing bandwidth allocation between \ac{D2D} and \ac{BS}-to-Device communication and the caching probability. The delay minimization problem is shown to be non-convex. Applying the block coordinate descent (BCD) optimization technique, the joint minimization problem is solved in an iterative manner.
\item We validate our theoretical findings via simulations. Results show a significant improvement in the network performance metrics, namely, the offloading gain, energy consumption, and average delay as compared to other caching schemes proposed earlier in literature. 
\end{itemize}

The rest of this paper is organized as follows. Section II and Section III discuss the system model and the offloading gain,  respectively. The energy consumption is discussed in Section IV and the delay analysis is conducted in Section V. Numerical results are then presented in Section VI before we conclude the paper in Section VII. 

\section{System Model}

\subsection{System Setup}
We model the location of the mobile devices with a \ac{TCP}  in which the parent points are drawn from a \ac{PPP} $\Phi_p$ with density $\lambda_p$, and the daughter points are drawn from a Gaussian \ac{PPP} around each parent point. In fact, the TCP is considered as a doubly \ac{PCP}  where the daughter points  are normally scattered with variance $\sigma^2 \in \mathbb{R}^2$ around each parent point \cite{haenggi2012stochastic}. 
 The parent points and offspring are referred to as cluster centers and cluster members, respectively. The number of cluster members in each cluster is a Poisson random variable with mean $\overline{n}$. The density function of the location of a cluster member relative to its cluster center is
\begin{equation}
f_Y(y) = \frac{1}{2\pi\sigma^2}\textrm{exp}\Big(-\frac{\lVert y\rVert^2}{2\sigma^2}\Big),	\quad \quad y \in \mathbb{R}^2
\label{pcp}
\end{equation}
where $\lVert .\rVert$ is the Euclidean norm. The intensity function of a cluster is given by $\lambda_c(y) = \frac{\overline{n}}{2\pi\sigma^2}\textrm{exp}\big(-\frac{\lVert y\rVert^2}{2\sigma^2}\big)$. Therefore, the intensity of the entire process is given by $\lambda = \overline{n}\lambda_p$. We assume that the BSs' distribution follows another \ac{PPP} $\Phi_{bs}$ with density $\lambda_{bs}$, which is independent of $\Phi_p$.

\subsection{Content Popularity and Probabilistic Caching Placement}
We assume that each device has a surplus memory of size $M$ designated for caching files. 
The total number of files is $N_f> M$ and the set (library) of content indices is denoted as $\mathcal{F} = \{1, 2, \dots , N_f\}$. These files represent the content catalog that all the devices in a cluster may request, which are indexed in a descending order of popularity. The probability that the $i$-th file is requested follows a Zipf's distribution given by,
\begin{equation}
q_i = \frac{ i^{-\beta} }{\sum_{k=1}^{N_f}k^{-\beta}},
\label{zipf}
\end{equation}
where $\beta$ is a parameter that reflects how skewed the popularity distribution is. For example, if $\beta= 0$, the popularity of the files has a uniform distribution. Increasing $\beta$ increases the disparity among the files popularity such that lower indexed files have higher popularity. By definition, $\sum_{i=1}^{N_f}q_i = 1$. 
We use Zipf's distribution to model the popularity of files per cluster.

\ac{D2D} communication is enabled within each cluster to deliver popular content. It is assumed that the devices adopt a slotted-ALOHA medium access protocol, where each transmitter during each time slot, independently and randomly accesses the channel with the same probability $p$. This implies that multiple active \ac{D2D} links might coexist within a cluster. Therefore, $p$ is a design parameter that directly controls \textcolor{black}{(mainly)} the intra-cluster interference, as described later in the paper. 

We adopt a random content placement where each device independently selects a file to cache according to a specific probability function $\textbf{b} = \{b_1, b_2, \dots, b_{N_{f}}\}$, where $b_i$ is the probability that a device caches the $i$-th file, $0 \leq b_i \leq 1$ for all $i=\{1, \dots, N_f\}$. To avoid duplicate caching of the same content within the memory of the same device, we follow a probabilistic caching approach proposed in \cite{blaszczyszyn2015optimal} and illustrated in Fig. \ref{prob_cache_example}.	
\begin{figure}
\vspace{-0.4cm}
	\begin{center}
		\includegraphics[width=1.5in]{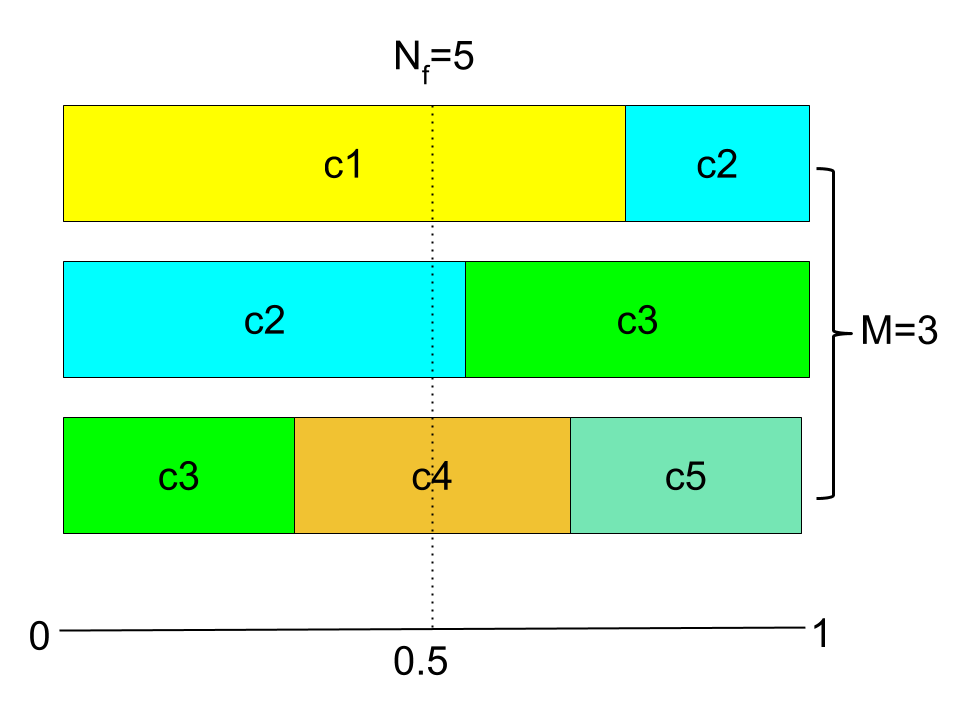} 
		\caption {The cache memory of size $M = 3$ is equally divided into $3$ blocks of unit size. A random number $\in [0,1]$ is generated, and a content $i$ is chosen from each block, whose $b_i$ fills the part intersecting with the generated random number. In this way, in the given example, the contents $\{1, 2, 4\}$ are chosen to be cached.}
		\label{prob_cache_example}
	\end{center}
\vspace{-0.8cm}
\end{figure}

If a device caches the desired file, the device directly retrieves the content. However, if the device does not cache the file, the file can be downloaded from any neighboring device that caches the file (henceforth called catering device) in the same cluster. According to the proposed access model, the probability that a chosen catering device is admitted to access the channel is the access probability $p$. Finally, the device attaches to the nearest \ac{BS} as a last resort to download the content which is not cached entirely within the device's cluster. We assume that the \ac{D2D} communication is operating as out-of-band D2D. \textcolor{black}{$W_{1}$ and $W_{2}$ denote respectively the bandwidth allocated to the \ac{D2D} and \ac{BS}-to-Device communication, and the total system bandwidth is denoted as $W=W_{1} + W_{2}$. It is assumed that device requests are served in a random manner, i.e., among the cluster devices, a random device request is chosen to be scheduled and content is served.}

In the following, we aim at studying and optimizing three important metrics, widely studied in the literature. The first metric is the offloading gain, which is defined as the probability of obtaining the requested file from the local cluster, either from the self-cache or from a neighboring device in the same cluster, with a rate higher than a required threshold $R_0$. The second metric is the energy consumption which represents the dissipated energy when downloading files either from the BSs or via \ac{D2D} communication. Finally, the latency which accounts for the weighted average delay over all the requests served from the \ac{D2D} and \ac{BS}-to-Device communication.  
\section{Maximum Offloading Gain} 

\begin{figure}[t]
\centering
\includegraphics[width=0.35\textwidth]{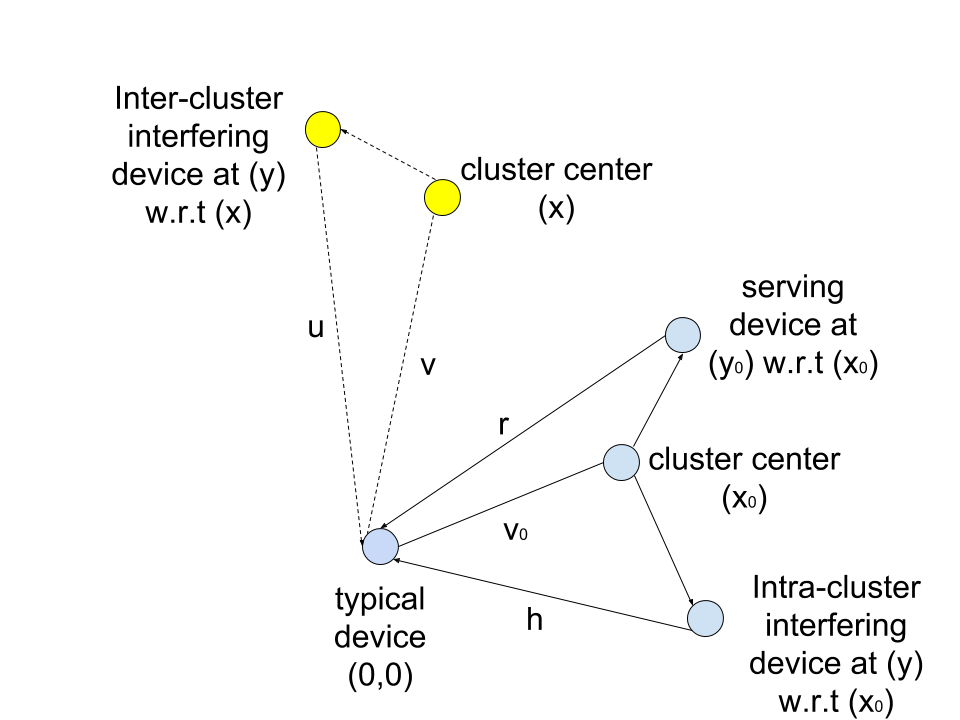}		
\vspace{-0.4cm}			
\caption {Illustration of the representative cluster and one interfering cluster.}
\label{distance}
\vspace{-0.5cm}
\end{figure}
Without loss of generality, we conduct the analysis for a cluster whose center is at $x_0\in \Phi_p$ (referred to as representative cluster), and the device who requests the content (henceforth called typical device) is located at the origin. We denote the location of the \ac{D2D} transmitter by $y_0$ relative to $x_0$, where $x_0, y_0\in \mathbb{R}^2$. The distance from the typical device (\ac{D2D} receiver of interest) to this \ac{D2D} transmitter is denoted as $r=\lVert x_0+y_0\rVert$, which is a realization of a random variable $R$ whose distribution is described later. This setup is illustrated in Fig. \ref{distance}. It is assumed that a requested file is served from a randomly selected catering device, which is, in turn, admitted to access the channel based on the slotted-ALOHA protocol. The successful offloading probability is then given by
 \begin{align}
 \label{offloading_gain}
 \mathbb{P}_o(p,\textbf{b}) &= \sum_{i=1}^{N_f} q_i b_i + q_i(1 - b_i)(1 - e^{-b_i\overline{n}})
 \underbrace{\int_{r=0}^{\infty}f_R(r) \mathbb{P}(R_{1}(r)>R_0) \dd{r}}_{\mathbb{P}(R_1>R_0)},
 \end{align}
where $R_{1}(r)$ is the achievable rate when downloading content from a catering device at a distance $r$ from the typical device with the distance \ac{PDF} $f_R(r)$. The first term on the right-hand side is the probability of requesting a locally cached file (self-cache) whereas the remaining term incorporates the probability that a requested file $i$ is cached among at least one cluster member and being downloadable with a rate greater than $R_0$. More precisely, since the number of devices per cluster has a Poisson distribution, the probability that there are $k$ devices per cluster is equal to $\frac{\overline{n}^k e^{-\overline{n}}}{k!}$. Accordingly, the probability that there are $k$ devices caching content $i$ can be written as $\frac{(b_i\overline{n})^k e^{-b_i\overline{n}}}{k!}$. Hence, the probability that at least one device caches content $i$ is $1$-minus the void probability (i.e., $k=0$), which equals $1 - e^{-b_i\overline{n}}$.

In the following, we first compute the probability $ \mathbb{P}(R_{1}(r)>R_0)$ given the distance $r$ between the typical device and a catering device, \textcolor{black}{then we conduct averaging over $r$ using the \ac{PDF} $f_R(r)$}. The received power at the typical device from a catering device located at $y_0$ relative to the cluster center is given by 
\begin{align}
P &= P_d  g_0 \lVert x_0+y_0\rVert^{-\alpha}= P_d  g_0 r^{-\alpha}			
\label{pwr}
\end{align}
where $P_d$ denotes the \ac{D2D} transmission power, \textcolor{black}{$g_0$ is the complex Gaussian fading channel coefficient between a catering device located at $y_0$ relative to its cluster center at $x_0$ and the typical device}, and $\alpha > 2$ is the path loss exponent. Under the above assumption, the typical device sees two types of interference, namely, the intra-and inter-cluster interference. We first describe the inter-cluster interference, then the intra-cluster interference is characterized. The set of active devices in any remote cluster is denoted as $\mathcal{B}^p$, where $p$ refers to the access probability. Similarly, the set of active devices in the local cluster is denoted as $\mathcal{A}^p$. Similar to (\ref{pwr}), the interference from the simultaneously active \ac{D2D} transmitters outside the representative cluster, at the typical device is given by
\begin{align}
I_{\Phi_p^{!}} &= \sum_{x \in \Phi_p^{!}} \sum_{y\in \mathcal{B}^p}   P_d g_{y_x}  \lVert x+y\rVert^{-\alpha}\\
& =  \sum_{x\in \Phi_p^{!}} \sum_{y\in \mathcal{B}^p}  P_d g_{u}  u^{-\alpha}
\end{align}
where $\Phi_p^{!}=\Phi_p \setminus x_0$ for ease of notation, $y$ is the marginal distance between a potential interfering device and its cluster center at $x \in \Phi_p$ , $u = \lVert x+y\rVert$ is a realization of a random variable $U$ modeling the inter-cluster interfering distance (shown in Fig. \ref{distance}), \textcolor{black}{$g_{y_x} \sim $ exp(1) are \ac{i.i.d.} exponential random variables modeling Rayleigh fading}, and $g_{u} = g_{y_x}$ for ease of notation. The intra-cluster interference is then given by
\begin{align}
I_{\Phi_c} &=  \sum_{y\in \mathcal{A}^p}   P_d g_{y_{x_0}}  \lVert x_0+y\rVert^{-\alpha}\\
& =   \sum_{y\in \mathcal{A}^p}  P_d g_{h}  h^{-\alpha}
\end{align}
where $y$ is the marginal distance between the intra-cluster interfering devices  and the cluster center at $x_0 \in \Phi_p$, $h = \lVert x_0+y\rVert$ is a realization of a random variable $H$ modeling the intra-cluster interfering distance (shown in Fig. \ref{distance}), \textcolor{black}{$g_{y_{x_0}} \sim $ exp(1) are \ac{i.i.d.} exponential random variables modeling Rayleigh fading, and $g_{h} = g_{y_{x_0}}$ for ease of notation}. From the thinning theorem \cite{haenggi2012stochastic}, the set of active transmitters following the slotted-ALOHA medium access forms Gaussian \ac{PPP} $\Phi_{cp}$ whose intensity is given by
 \begin{align}
\lambda_{cp} = p\lambda_{c}(y) = p\overline{n}f_Y(y) =\frac{p\overline{n}}{2\pi\sigma^2}\textrm{exp}\Big(-\frac{\lVert y\rVert^2}{2\sigma^2}\Big),	\quad		y \in \mathbb{R}^2
 \end{align}
 Assuming that the thermal noise is neglected as compared to the aggregate interference, the \ac{D2D} $\sir$ at the typical device is written as 
\begin{equation}
\gamma_{r}=\frac{P}{I_{\Phi_p^{!}} + I_{\Phi_c}} = \frac{P_d  g_0 r^{-\alpha}}{I_{\Phi_p^{!}} + I_{\Phi_c}}
\end{equation}
\textcolor{black}{A fixed rate transmission model is adopted in our study}, where each transmitter (device or BS) transmits at the fixed rate of log$_2[1+\theta]$ \SI{}{bits/sec/Hz}, where $\theta$ is a design parameter. Since, the rate is fixed, the transmission is subject to outage due to fading and interference fluctuations. Consequently, the de facto average transmissions rate (i.e., average throughput) is given by 
\begin{equation}
\label{rate_eqn}
R_i = W_i\textrm{ log$_{2}$}[1+ \theta]\pc,			
\end{equation}
\textcolor{black}{where $i= 1,2$ for the \ac{D2D} and \ac{BS}-to-Device communication, respectively. $W_i$ is the bandwidth, $\theta$ is the pre-determined threshold for successful reception, $\pc =\mathbb{E}(\textbf{1}\{\sir>\theta\})$ is the coverage probability, and $\textbf{1}\{.\}$ is the indicator function. When served by a catering device $r$ apart from the origin, the achievable rate of the typical device under slotted-ALOHA medium access technique can be deduced from \cite[Equation (10)]{jing2012achievable} as}
 \begin{equation}
\label{rate_eqn1}
R_{1}(r) = pW_{1} {\rm log}_2 \big(1 + \theta \big) \textbf{1}\{ \gamma_{r} > \theta\}
\end{equation}
Then, the probability $ \mathbb{P}(R_{1}(r)>R_0)$ is derived as follows.
\begin{align}
\mathbb{P}(R_{1}(r)>R_0)&= \mathbb{P} \Big(pW_{1} {\rm log}_2 (1 + \theta)\textbf{1}\{ \gamma_{r} > \theta\}
>R_0\Big)	\nonumber  \\
&=\mathbb{P} \Big(\textbf{1}\{ \gamma_{r} > \theta\}
>\frac{R_0}{pW_{1} {\rm log}_2 (1 + \theta )}\Big) \nonumber  \\
&\overset{(a)}{=}  \mathbb{P}\big(\gamma_r >\theta \big)=\mathbb{P}\Big(\frac{P_d  g_0 r^{-\alpha}}{I_{\Phi_p^{!}} + I_{\Phi_c}} > \theta\Big)
\end{align}
\textcolor{black}{where (a) follows from the assumption that $R_0 < pW_{1} {\rm log}_2 \big(1 + \theta \big)$, i.e., $\frac{R_0}{pW_{1} {\rm log}_2 \big(1 + \theta \big)}<1$, always holds, otherwise, it is infeasible to get $\mathbb{P}(R_{1}>R_0)$ greater than zero}. Rearranging the right-hand side, we get
\begin{align}
\mathbb{P}(R_{1}(r)>R_0)&= \mathbb{P}\Big( g_0 >  \frac{\theta r^{\alpha}}{P_d} [I_{\Phi_p^{!}} + I_{\Phi_c}]  \Big)
\nonumber  \\
&\overset{(b)}{=}\mathbb{E}_{I_{\Phi_p^{!}},I_{\Phi_c}}\Big[\text{exp}\big(\frac{-\theta r^{\alpha}}{P_d}{[I_{\Phi_p^{!}} + I_{\Phi_c}] }\big)\Big]
\nonumber  \\
\label{prob-R1-g-R0}
 &\overset{(c)}{=}  \mathscr{L}_{I_{\Phi_p^{!}}}\big(s=\frac{\theta r^{\alpha}}{P_d}\big) \mathscr{L}_{I_{\Phi_c}} \big(s=\frac{\theta r^{\alpha}}{P_d}\big)
\end{align}
where (b) follows from the assumption $g_0 \sim \mathcal{CN}(0,1)$, and (c) follows from the independence of the intra- and inter-cluster interference and the Laplace transform of them. 
In what follows, we first derive the Laplace transform of interference to get $\mathbb{P}(R_{1}(r)>R_0)$. Then, we formulate the offloading gain maximization problem.
\begin{lemma}
The Laplace transform of the inter-cluster aggregate interference $I_{\Phi_p^{!}}$ evaluated at $s=\frac{\theta r^{\alpha}}{P_d}$ is given by 
\begin{align}
 \label{LT_inter}
\mathscr{L}_{I_{\Phi_p^{!}}}(s) &= {\rm exp}\Big(-2\pi\lambda_p \int_{v=0}^{\infty}\Big(1 -  {\rm e}^{-p\overline{n} \varphi(s,v)}\Big)v\dd{v}\Big),
\end{align}
where $\varphi(s,v) = \int_{u=0}^{\infty}\frac{s}{s+ u^{\alpha}}f_U(u|v)\dd{u}$, and $f_U(u|v) = \mathrm{Rice} (u| v, \sigma)$ represents Rice's \ac{PDF} of parameter $\sigma$, and $v=\lVert x\rVert$.
\end{lemma}
%
\begin{proof}
Please see Appendix A.
\end{proof}

\begin{lemma}
The Laplace transform of the intra-cluster aggregate interference $I_{\Phi_c}$ evaluated at $s=\frac{\theta r^{\alpha}}{P_d}$ can be approximated by
\begin{align}   
\label{LT_intra}
         \mathscr{L}_{I_{\Phi_c} }(s) \approx  {\rm exp}\Big(-p\overline{n} \int_{h=0}^{\infty}\frac{s}{s+ h^{\alpha}}f_H(h)\dd{h}\Big)
\end{align} 
where $f_H(h) =\mathrm{Rayleigh}(h,\sqrt{2}\sigma)$ represents Rayleigh's \ac{PDF} with a scale parameter $\sqrt{2}\sigma$. 
\end{lemma}
\begin{proof}
Please see Appendix B.
\end{proof}

For the serving distance distribution $f_R(r)$, since both the typical device as well as a potential catering device have their locations drawn from a normal distribution with variance $\sigma^2$ around the cluster center, then by definition, the serving distance has a Rayleigh distribution with parameter $\sqrt{2}\sigma$, and given by
  \begin{align}
\label{rayleigh}
f_R(r)=  \frac{r}{2 \sigma^2} {\rm e}^{\frac{-r^2}{4 \sigma^2}}, 	\quad	r>0
 \end{align}

From (\ref{LT_inter}), (\ref{LT_intra}), and (\ref{rayleigh}), the offloading gain in (\ref{offloading_gain}) is written as 
 \begin{align}
 \label{offloading_gain_1}
 \mathbb{P}_o(p,\textbf{b}) &= \sum_{i=1}^{N_f} q_i b_i + q_i(1 - b_i)(1 - e^{-b_i\overline{n}}) \underbrace{\int_{r=0}^{\infty} 
 \frac{r}{2 \sigma^2} {\rm e}^{\frac{-r^2}{4 \sigma^2}} 
 \mathscr{L}_{I_{\Phi_p^{!}}}\big(s=\frac{\theta r^{\alpha}}{P_d}\big) \mathscr{L}_{I_{\Phi_c}} \big(s=\frac{\theta r^{\alpha}}{P_d}\big)\dd{r}}_{\mathbb{P}(R_1>R_0)},				
 \end{align}
Hence, the offloading gain maximization problem can be formulated as
\begin{align}
\label{optimize_eqn_p}
\textbf{P1:}		\quad &\underset{p,\textbf{b}}{\text{max}} \quad \mathbb{P}_o(p,\textbf{b}) \\
\label{const110}
&\textrm{s.t.}\quad  \sum_{i=1}^{N_f} b_i = M,   \\
\label{const111}
&  b_i \in [ 0, 1],  \\
\label{const112}
&  p \in [ 0, 1], 
\end{align}	
where (\ref{const110}) is the device cache size constraint, which is consistent with the illustration of the example in Fig. \ref{prob_cache_example}. \textcolor{black}{On one hand, from the assumption that the fixed transmission rate $pW_{1} {\rm log}_2 \big(1 + \theta \big)$ being larger than the required threshold $R_0$, we have the condition $p>\frac{R_0}{W_{1} {\rm log}_2 \big(1 + \theta \big)}$ on the access probability. On the other hand, from (\ref{prob-R1-g-R0}), with further increase of the access probability $p$, intra- and inter-cluster interference powers increase, and the probability $\mathbb{P}(R_{1}(r)>R_0)$ decreases accordingly. From intuition, the optimal access probability for the offloading gain maximization is chosen as $p^* >_{\epsilon} \frac{R_0}{W_{1} {\rm log}_2 \big(1 + \theta \big)}$, where $\epsilon \to 0$. However, increasing the access probability $p$ further above $p^*$ may lead to higher \ac{D2D} average achievable rate $R_1$, as elaborated in the next section.} The obtained $p^*$ is now used to solve for the caching probability $\textbf{b}$ in the optimization problem below. Since in the structure of \textbf{P1}, $p$ and $\textbf{b}$ are separable, it is possible to solve for $p^*$ and then substitute to get $\textbf{b}^*$.
\begin{align}
\label{optimize_eqn_b_i}
\textbf{P2:}		\quad &\underset{\textbf{b}}{\text{max}} \quad \mathbb{P}_o(p^*,\textbf{b}) \\
\label{const11}
&\textrm{s.t.}\quad  (\ref{const110}), (\ref{const111})   \nonumber 
\end{align}
The optimal caching probability is formulated in the next lemma.	
\begin{lemma}
$\mathbb{P}_o(p^*,\textbf{b})$ is a concave function w.r.t. $\textbf{b}$ and the optimal caching probability $\textbf{b}^{*}$ that maximizes the offloading gain is given by
      \[
    b_{i}^{*}=\left\{
                \begin{array}{ll}
                  1 \quad\quad\quad , v^{*} <   q_i -q_i(1-e^{-\overline{n}})\mathbb{P}(R_1>R_0)\\ 
                  0   \quad\quad\quad, v^{*} >   q_i + \overline{n}q_i\mathbb{P}(R_1>R_0)\\
                 \psi(v^{*}) \quad, {\rm otherwise}
                \end{array}   
              \right.
  \]
where $\psi(v^{*})$ is the solution of $v^{*}$ of (\ref{psii_offload}) in Appendix C that satisfies $\sum_{i=1}^{N_f} b_i^*=M$.
\end{lemma}
\begin{proof}
Please see Appendix C.
\end{proof}
\section{Energy Consumption}
In this section, we formulate the energy consumption minimization problem for the clustered \ac{D2D} caching network. In fact, significant energy consumption occurs only when content is served via \ac{D2D} or \ac{BS}-to-Device transmission. We consider the time cost $c_{d_i}$ as the time it takes to download the $i$-th content from a neighboring device in the same cluster. Considering the size ${S}_i$ of the $i$-th ranked content, $c_{d_i}=S_i/R_1 $, where $R_1 $ denotes the average rate of the \ac{D2D} communication. Similarly, we have $c_{b_i} = S_i/R_2 $ when the $i$-th content is served by the \ac{BS} with average rate $R_2 $. The average energy consumption when downloading files by the devices in the representative cluster is given by
\begin{align}
\label{energy_avrg}
E_{av} = \sum_{k=1}^{\infty} E(\textbf{b}|k)\mathbb{P}(n=k)
\end{align}
where $\mathbb{P}(n=k)$ is the probability that there are $k$ devices in the representative cluster, \textcolor{black}{equal to $\frac{\overline{n}^k e^{-\overline{n}}}{k!}$}, and $E(\textbf{b}|k)$ is the energy consumption conditioning on having $k$ devices in the cluster, written similar to \cite{energy_efficiency} as
\begin{equation}
E(\textbf{b}|k) = \sum_{j=1}^{k} \sum_{i=1}^{N_f}\big[\mathbb{P}_{j,i}^d q_i P_d c_{d_i} + \mathbb{P}_{j,i}^b q_i P_b c_{b_i}\big]
\label{energy}
\end{equation}
where $\mathbb{P}_{j,i}^d$ and $\mathbb{P}_{j,i}^b$ represent the probability of obtaining the $i$-th content by the $j$-th device from the local cluster, i.e., via \ac{D2D} communication, and the BS, respectively. $P_b$ denotes the \ac{BS} transmission power. Given that there are $k$ devices in the cluster, it is obvious that $\mathbb{P}_{j,i}^b=(1-b_i)^{k}$, and $\mathbb{P}_{j,i}^d=(1 - b_i)\big(1-(1-b_i)^{k-1}\big)$. 

The average rates $R_1$ and $R_2$ are now computed to get a closed-form expression for $E(\textbf{b}|k)$. 
From equation (\ref{rate_eqn}), we need to obtain the \ac{D2D} coverage probability $\pcd$ and \ac{BS}-to-Device coverage probability $\pcb$ to calculate $R_1$ and $R_2$, respectively. Given the number of devices $k$ in the representative cluster, the Laplace transform of the inter-cluster interference \textcolor{black}{is as obtained in (\ref{LT_inter})}. However, the intra-cluster interfering devices no longer represent a Gaussian \ac{PPP} since the number of devices is conditionally fixed, i.e., not a Poisson random number as before. To facilitate the analysis, for every realization $k$, we assume that the intra-cluster interfering devices form a Gaussian \ac{PPP} with intensity function given by $pkf_Y(y)$. Such an assumption is mandatory for analytical tractability. 
From Lemma 2, the intra-cluster Laplace transform conditioning on $k$ can be approximated as
\begin{align}
\mathscr{L}_{I_{\Phi_c} }(s|k) &\approx  {\rm exp}\Big(-pk \int_{h=0}^{\infty}\frac{s}{s+ h^{\alpha}}f_H(h)\dd{h}\Big) 
\nonumber
 \end{align}
and the conditional \ac{D2D} coverage probability is given by
\begin{align}
\label{p_b_d2d}
\pcd = 
\int_{r=0}^{\infty} \frac{r}{2 \sigma^2} {\rm e}^{\frac{-r^2}{4 \sigma^2}} 
 \mathscr{L}_{I_{\Phi_p^{!}}}\big(s=\frac{\theta r^{\alpha}}{P_d}\big) \mathscr{L}_{I_{\Phi_c}} \big(s=\frac{\theta r^{\alpha}}{P_d}\big|k\big)\dd{r}
 \end{align}
\textcolor{black}{With the adopted slotted-ALOHA scheme, the access probability $p$ minimizing $E(\textbf{b}|k)$  is computed over the interval [0,1] to maximize the \ac{D2D} achievable rate $R_1$ in (\ref{rate_eqn1}), with the condition $p>\frac{R_0}{W_{1} {\rm log}_2 \big(1 + \theta \big)}$ holding to fulfill the probability $\mathbb{P}(R_{1}>R_0)$ greater than zero.
As an illustrative example, in Fig.~\ref{prob_r_geq_r0_vs_p}, we plot the \ac{D2D} average achievable rate $R_1$ against the channel access probability $p$. 
As evident from the plot, we see that there is a certain access probability, such that before it the rate $R_1$ tends to increase since the channel access probability increases, and beyond it, the rate $R_1$ decreases monotonically due to the effect of more interferers accessing the channel. In such a case, although we observe that increasing $p$ above $\frac{R_0}{W_{1} {\rm log}_2 \big(1 + \theta \big)}=0.1$ improves the average achievable rate $R_1$, it comes at a price of a decreased $\mathbb{P}(R_{1}>R_0)$.} 
\begin{figure}[!h]
	\begin{center}
		\includegraphics[width=2.5in]{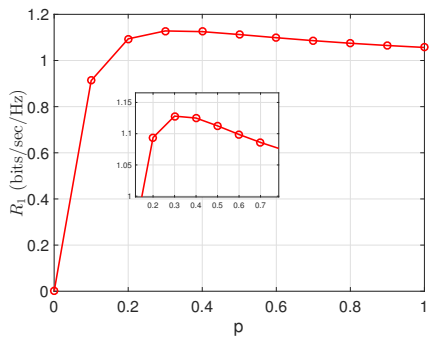}	
		\caption {\textcolor{black}{The \ac{D2D} average achievable rate $R_1$ versus the access probability $p$ ($\lambda_p = 20 \text{ clusters}/\SI{}{km}^2$, $\overline{n}=12$, $\sigma=\SI{30}{m}$, $\theta=\SI{0}{dB}$, $R_0/W_1=0.1\SI{}{bits/sec/Hz}$).}}
		\label{prob_r_geq_r0_vs_p}
	\end{center}
\vspace{-0.6cm}
\end{figure}

Analogously, under the \ac{PPP} $\Phi_{bs}$, and based on the nearest \ac{BS} association principle, it is shown in \cite{andrews2011tractable} that the \ac{BS} coverage probability can be expressed as
\begin{equation}
\pcb =\frac{1}{{}_2 F_1(1,-\delta;1-\delta;-\theta)},
\label{p_b_bs}
\end{equation}
where ${}_2 F_1(.)$ is the Gaussian hypergeometric function and $\delta = 2/\alpha$. Given the coverage probabilities $\pcd$ and $\pcb$ in (\ref{p_b_d2d}) and (\ref{p_b_bs}), respectively, $R_1$ and $R_2 $ can be calculated from (\ref{rate_eqn}), and hence $E(\textbf{b}|k)$ is expressed in a closed-form.           
\subsection{Energy Consumption Minimization}
The energy minimization problem can be formulated as
\begin{align}
\label{optimize_eqn1}
&\textbf{P3:}		\quad\underset{\textbf{b}}{\text{min}} \quad E(\textbf{b}|k) = \sum_{j=1}^{k} \sum_{i=1}^{N_f}\big[\mathbb{P}_{j,i}^d q_i P_d c_{d_i} + \mathbb{P}_{j,i}^b q_i P_b c_{b_i}\big]
\\
\label{const11}
&\textrm{s.t.}\quad  (\ref{const110}), (\ref{const111})   \nonumber 
\end{align}	
In the next lemma, we prove the convexity condition for $E(\textbf{b}|k)$.
\begin{lemma}
\label{convex_E}
The energy consumption $E(\textbf{b}|k)$ is convex if $\frac{P_b}{R_2}>\frac{P_d}{R_1}$.
\end{lemma}
\begin{proof}
\textcolor{black}{We proceed by deriving the Hessian matrix of $E(\textbf{b}|k)$. The Hessian matrix of $E(\textbf{b}|k)$ w.r.t. the caching variables is \textbf{H}$_{i,j} = \frac{\partial^2 E(\textbf{b}|k)}{\partial b_i \partial b_j}$, $\forall i,j \in \mathcal{F}$. \textbf{H}$_{i,j}$ a diagonal matrix whose $i$-th row and $j$-th column element is given by $k(k-1) S_i\Big(\frac{P_b}{R_2}-\frac{P_d}{R_1}\Big)q_i(1 - b_i)^{k-2}$.}
Since the obtained Hessian matrix is full-rank and diagonal, $\textbf{H}_{i,j}$ is positive semidefinite (and hence $E(\textbf{b}|k)$ is convex) if all the diagonal entries are nonnegative, i.e., when
$\frac{P_b}{R_2}>\frac{P_d}{R_1}$. In practice, it is reasonable to assume that  $P_b \gg P_d$, as in \cite{ericsson}, the \ac{BS} transmission power is 100 fold the \ac{D2D} power.
\end{proof}
As a result of Lemma 3, the optimal caching probability can be computed to minimize $E(\textbf{b}|k)$.
\begin{lemma}
The optimal caching probability $\textbf{b}^{*}$ for the energy minimization problem \textbf{P3} is given by,
\begin{equation}
b_i^* = \Bigg[ 1 - \Big(\frac{v^* + k^2q_iS_i\frac{P_d}{R_1 }}{kq_iS_i\big(\frac{P_d}{R_1 }-\frac{P_b}{R_2}\big)} \Big)^{\frac{1}{k-1}}\Bigg]^{+}
\label{energy11}
\end{equation} 
 where $v^{*}$ satisfies the maximum cache constraint $\sum_{i=1}^{N_f} \Big[ 1 - \Big(\frac{v^* + k^2q_iS_i\frac{P_d}{R_1 }}{kq_iS_i\big(\frac{P_d}{R_1 }-\frac{P_b}{R_2 }\big)} \Big)^{\frac{1}{k-1}}\Big]^{+}=M$, and $[x]^+ =$ max$(x,0)$.
\end{lemma}
\begin{proof}
The proof proceeds in a similar manner to Lemma 3 and is omitted.
\end{proof}
\begin{proposition} {\rm \textcolor{black}{By observing (\ref{energy11}), we can demonstrate the effects of content size and popularity on the optimal caching probability. $S_i$ exists in the numerator and denominator
of the second term in (\ref{energy11}), however, the effect on numerator is more significant due to larger multiplier. The same property is observed for $q_i$. With the increase of $S_i$ or $q_i$, the magnitude of the second term in (\ref{energy11}) increases, and correspondingly, $b_i^*$ decreases. That is a content with larger size or lower popularity has smaller probability to be cached.}}   
\end{proposition}

By substituting $b_i^*$ into (\ref{energy_avrg}), the average energy consumption per cluster is obtained. In the rest of the paper, we study and minimize the weighted average delay per request for the proposed system.
\section{Delay Analysis}
In this section, the delay analysis and minimization are discussed. A joint stochastic geometry and queueing theory model is exploited to study this problem. The delay analysis incorporates the study of a system of spatially interacting queues. To simplify the mathematical analysis, we further consider that only one \ac{D2D} link can be active within a cluster of $k$ devices, where $k$ is fixed. As shown later, such an assumption facilitates the analysis by deriving simple expressions. We begin by deriving the \ac{D2D} coverage probability under the above assumption, which is used later in this section.
\begin{lemma}
\label{coverage_lemma}
The \ac{D2D} coverage probability of the proposed clustered model with one active \ac{D2D} link within a cluster is given by
 \begin{align} 
 \label{v_0} 
 \pcd = \frac{1}{4\sigma^2 Z(\theta,\alpha,\sigma)} , 
 \end{align}
where $Z(\theta,\alpha,\sigma) = (\pi\lambda_p \theta^{2/\alpha}\Gamma(1 + 2/\alpha)\Gamma(1 - 2/\alpha)+ \frac{1}{4\sigma^2})$.
\end{lemma}			
\begin{proof}
The result can be proved by using the displacement theory of the \ac{PPP} \cite{daley2007introduction}, and then proceeding in a similar manner to Lemma 1 and 2. We delegate this proof to the conference version of this paper \cite{amer2018minimizing}.
\end{proof}
In the following, we firstly describe the traffic model of the network, and then we formulate the delay minimization problem. 
\begin{figure}
	\begin{center}
		\includegraphics[width=2.5in]{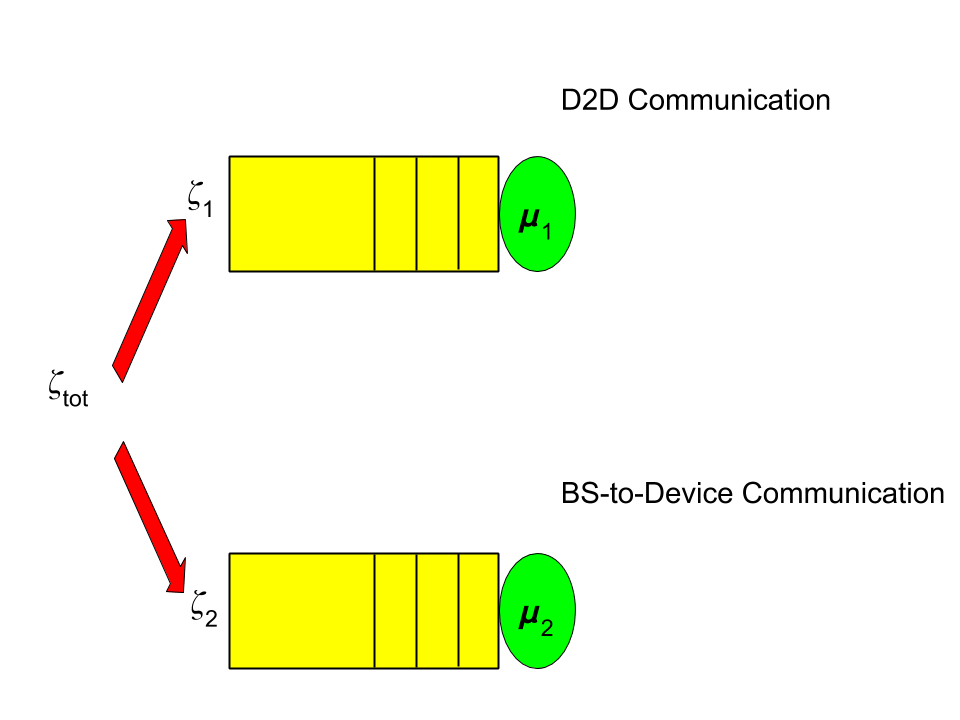}		
		\caption {\textcolor{black}{The traffic model of request arrivals and departures in a given cluster. $Q_1$ and $Q_2$ are M/G/1 queues modeling requests served by \ac{D2D} and \ac{BS}-to-Device communication, respectively.}}
		\label{delay_queue}
	\end{center}
\vspace{-0.8cm}
\end{figure}

\subsection{Traffic Model}
We assume that the aggregate request arrival process from the devices in each cluster follows a Poisson arrival process with parameter $\zeta_{tot}$ (requests per time slot). As shown in Fig.~\ref{delay_queue}, the incoming requests are further divided according to where they are served from.  $\zeta_{1}$ represents the arrival rate of requests served via the \ac{D2D} communication, whereas $\zeta_{2}$ is the arrival rate for those served from the BSs. $\zeta_{3} = 1 - \zeta_{1} - \zeta_{2}$  denotes the arrival rate of requests served via the self-cache with zero delay. By definition, $\zeta_{1}$ and $\zeta_{2}$ are also Poisson arrival processes. Without loss of generality, we assume that the file size has a general distribution $G$ whose mean is denoted as $\overline{S}$ \SI{}{MBytes}. Hence, an M/G/1 queuing model is adopted whereby two non-interacting queues, $Q_1$ and $Q_2$, model the traffic in each cluster served via the \ac{D2D} and \ac{BS}-to-Device communication, respectively. Although $Q_1$ and $Q_2$ are non-interacting as the \ac{D2D} communication is assumed to be out-of-band, these two queues are  spatially interacting with similar queues in other clusters. To recap, $Q_1$ and $Q_2$ are two M/G/1 queues with arrival rates $\zeta_{1}$ and $\zeta_{1}$, and service rates $\mu_1$ and $\mu_2$, respectively. 

\subsection{Queue Dynamics}
It is worth highlighting that the two queues $Q_i$, $i \in \{1,2\}$, accumulate requests for files demanded by the clusters members, not the files themselves. First-in first-out (FIFO) scheduling is assumed where a request for content arrives first will be scheduled first either by the \ac{D2D} or \ac{BS} communication if the content is cached among the devices or not, respectively. The result of FIFO scheduling only relies on the time when the request arrives at the queue and is irrelevant to the particular device that issues the request. Given the parameter of the Poisson's arrival process $\zeta_{tot}$, the arrival rates at the two queues are expressed respectively as
\begin{align}
\zeta_{1} &= \zeta_{tot} \sum_{i=1}^{N_f}q_i\big((1 - b_i) -(1-b_i)^{k}\big),	\\
\zeta_{2} &= \zeta_{tot} \sum_{i=1}^{N_f}q_i (1-b_i)^{k} 
\end{align}

The network operation is depicted in Fig. \ref{delay_queue}, and described in detail below.
\begin{enumerate}

\item Given the memoryless property of the arrival process (Poisson arrival) along with the assumption that the service process is independent of the arrival process,
the number of requests in any queue at a future time only depends upon the current number in the system (at time $t$) and the arrivals or departures that occur within the interval $e$.
\begin{align}
Q_{i}(t+e) =  Q_{i}(t) +   \Lambda_{i}(e)  -   M_{i}(e)  
\end{align}
where $\Lambda_{i}(e)$ is the number of arrivals in the time interval $(t,t+e)$, whose mean is $\zeta_i$ \SI{}{sec}$^{-1}$, and $M_{i}(e)$ is the number of departures in the time interval $(t,t+e)$, whose mean is $\mu_i = \frac{\mathbb{E}(\textbf{1} \{\sir>\theta\})W_i{\rm log}_2(1+\theta)}{\overline{S}}$ \SI{}{sec}$^{-1}$. It is worth highlighting that, unlike the spatial-only model studied in the previous sections, the term $\mathbb{E}(\textbf{1} \{\sir>\theta\})$ is dependent on the traffic dynamics \textcolor{black}{since a request being served in a given cluster is interfered only from other clusters that also have requests to serve}. What is more noteworthy is that the mean service time $\tau_i = \frac{1}{\mu_i}$ follows the same distribution as the file size. These aspects will be revisited  later in this section.

\item $\Lambda_{i}(e)$ is dependent only on $e$ because the arrival process is Poisson. $M_{i}(e)$ is $0$  if the service time of the file being served $\epsilon_i >e$. $M_{i}(e)$ is 1 if $\epsilon_1 <e$ and $\epsilon_2 + \epsilon_1>e$, and so on. As the service times $\epsilon_1, \epsilon_2, \dots, \epsilon_n$ are independent, neither $\Lambda_{i}(e)$ nor $M_{i}(e)$ depends on what happened prior to $t$. Thus, $Q_{i}(t+e)$ only depends upon $Q_{i}(t)$ and not the past history. Hence it is a \ac{CTMC} which obeys the stability conditions in \cite{szpankowski1994stability}.
\end{enumerate}
The following proposition provides the sufficient conditions for the stability of the buffers in the sense defined in \cite{szpankowski1994stability}, i.e., $\{Q_{i}\}$ has a limiting distribution for $t \rightarrow \infty$.

\begin{proposition} {\rm The \ac{D2D} and \ac{BS}-to-Device traffic modeling queues are stable, respectively, if and only if}
\begin{align}
\label{stable1}
\zeta_{1} < \mu_1 &= \frac{\pcd W_1{\rm log}_2(1+\theta)}{\overline{S}} \\
\label{stable2}
\zeta_{2} < \mu_2 & =\frac{\pcb W_2{\rm log}_2(1+\theta)}{\overline{S}} 
\end{align}
\end{proposition}

\begin{proof}
We  show sufficiency by proving that (\ref{stable1}) and (\ref{stable2}) guarantee stability in a dominant network, where all queues that have empty buffers make dummy transmissions. The dominant network is a fictitious system that is identical to the original system, except that terminals may choose to transmit even when their respective buffers are empty, in which case they simply transmit a dummy packet. If both systems are started from the same initial state and fed with the same arrivals, then the queues in the fictitious dominant system can never be shorter than the queues in the original system. 
\textcolor{black}{Similar to the spatial-only network, in the dominant system, the typical receiver is seeing an interference from all other clusters whether they have requests to serve or not (dummy transmission).} This dominant system approach yields $\mathbb{E}(\textbf{1} \{\sir>\theta\})$ equal to $\pcd$ and $\pcb$ for the \ac{D2D} and \ac{BS}-to-Device communication, respectively. Also, the obtained delay is an upper bound for the actual delay of the system. The necessity of (\ref{stable1}) and (\ref{stable2}) is shown as follows: If $\zeta_i>\mu_i$, then, by Loynes' theorem \cite{loynes1962stability}, it follows that lim$_{t\rightarrow \infty}Q_i(t)=\infty$ (a.s.) for all queues in the dominant network. 
\end{proof}

Next, we conduct the analysis for the dominant system whose parameters are as follows. The content size has an exponential distribution of mean $\overline{S}$ \SI{}{MBytes}. The service times also obey an exponential distribution with means $\tau_1 = \frac{\overline{S}}{R_1}$ \SI{}{seconds} and $\tau_2 = \frac{\overline{S}}{R_2}$ \SI{}{seconds}. The rates $R_1$ and $R_2$ are calculated from (\ref{rate_eqn}) where $\pcd$ and $\pcb$ are from (\ref{v_0}) and (\ref{p_b_bs}), respectively. Accordingly, $Q_1$ and $Q_2$ are two continuous time independent (non-interacting) M/M/1 queues with service rates $\mu_1 = \frac{\pcd W_1{\rm log}_2(1+\theta)}{\overline{S}}$ and $\mu_2 =  \frac{\pcb W_2{\rm log}_2(1+\theta)}{\overline{S}}$ \SI{}{sec}$^{-1}$, respectively. \begin{proposition} {\rm The mean queue length $L_i$ of the $i$-th queue is given by} 
\begin{align}
\label{queue_len}
L_i &=  \rho_i + \frac{2\rho_i^2}{2\zeta_i(1 - \rho_i)},      
\end{align}
\end{proposition}
\begin{proof}
We can easily calculate $L_i$ by observing that $Q_i$ are continuous time M/M/1 queues with arrival rates $\zeta_i$, service rates $\mu_i$, and traffic intensities $\rho_i = \frac{\zeta_i}{\mu_i}$. Then, by applying the Pollaczek-Khinchine formula \cite{Kleinrock}, $L_i$ is directly obtained. 
\end{proof}
The average delay per request for each queue is calculated from
\begin{align}
D_1 &= \frac{L_1}{\zeta_1}= \frac{1}{\mu_1 - \zeta_1} = \frac{1}{W_1\mathcal{O}_{1} - \zeta_{tot} \sum_{i=1}^{N_f}q_i\big((1 - b_i) -(1-b_i)^{k}\big)} \\
D_2 &=  \frac{L_2}{\zeta_2}=\frac{1}{\mu_2 - \zeta_2} = \frac{1}{W_2\mathcal{O}_{2} - \zeta_{tot} \sum_{i=1}^{N_f}q_i (1-b_i)^{k}} 
\end{align}
where $\mathcal{O}_{1} = \frac{\pcd {\rm log}_2(1+\theta)}{\overline{S}}$, $\mathcal{O}_{2}= \frac{\pcb {\rm log}_2(1+\theta)}{\overline{S}}$ for notational simplicity. The weighted average delay $D$ is then expressed as
\begin{align}
D&=  \frac{\zeta_{1}D_1 + \zeta_{2}D_2}{\zeta_{tot}}  \nonumber \\
&= \frac{\sum_{i=1}^{N_f}q_i\big((1 - b_i) -(1-b_i)^{k}\big)}{ \mathcal{O}_{1}W_1 - \zeta_{tot} \sum_{i=1}^{N_f}q_i\big((1 - b_i) -(1-b_i)^{k}\big)} + \frac{\sum_{i=1}^{N_f}q_i (1-b_i)^{k}}{ \mathcal{O}_{2}W_2 - \zeta_{tot} \sum_{i=1}^{N_f}q_i (1-b_i)^{k}}
\end{align}
One important insight from the delay equation is that the caching probability $\textbf{b}$ controls the arrival rates $\zeta_{1}$ and $\zeta_{2}$ while the bandwidth determines the service rates $\mu_1$ and $\mu_2$. Therefore, it turns out to be of paramount importance to jointly optimize $\textbf{b}$ and $W_1$ to minimize the average delay. One relevant work is carried out in \cite{tamoor2016caching} where the authors investigate the storage-bandwidth tradeoffs for small cell \ac{BS}s that are subject to storage constraints. Subsequently, we formulate the weighted average delay joint caching and bandwidth minimization problem as
\begin{align}
\label{optimize_eqn3}
\textbf{P4:} \quad \quad&\underset{\textbf{b},{\rm W}_1}{\text{min}} \quad D(\textbf{b},W_1) \\
&\textrm{s.t.}\quad  (\ref{const110}), (\ref{const111})   \nonumber \\
&  0 \leq W_1 \leq W,  \\
\label{stab1}
&\zeta_{tot} \sum_{i=1}^{N_f}q_i\big((1 - b_i) -(1-b_i)^{k}\big) < \mu_1,	\\
\label{stab2}
&\zeta_{tot} \sum_{i=1}^{N_f}q_i (1-b_i)^{k} < \mu_2, 
\end{align}
where constraints (\ref{stab1}) and (\ref{stab2}) are the stability conditions for the queues $Q_1$ and $Q_2$, respectively. Although the objective function of \textbf{P4} is convex w.r.t. $W_1$, as derived below, the coupling of the optimization variables $\textbf{b}$ and $W_1$ makes \textbf{P4} a non-convex optimization problem. Therefore, \textbf{P4} cannot be solved directly using standard convex optimization techniques. 
\textcolor{black}{By applying the \ac{BCD} optimization technique, \textbf{P4} can be solved in an iterative manner as follows. First, for a given caching probability $\textbf{b}$, we calculate the bandwidth allocation subproblem. Afterwards, the obtained optimal bandwidth is used to update $\textbf{b}$}. The optimal bandwidth for the bandwidth allocation subproblem is given in the next Lemma. 
\begin{lemma}
The objective function of \textbf{P4} in (\ref{optimize_eqn3}) is convex w.r.t. $W_1$, and the optimal bandwidth allocation to the \ac{D2D} communication is given by
\begin{align}
\label{optimal-w-1}
W_1^* = \frac{\zeta_{tot}\sum_{i=1}^{N_f}q_i(\overline{b}_i - \overline{b}_i^{k}) +\varpi \big(\mathcal{O}_{2}W - \zeta_{tot}\sum_{i=1}^{N_f}q_i\overline{b}_i^{k}\big)}{\mathcal{O}_{1}+\varpi\mathcal{O}_{2}},
\end{align}
where $\overline{b}_i = 1 - b_i$ and $\varpi=\sqrt{\frac{\mathcal{O}_{1}\sum_{i=1}^{N_f}q_i(\overline{b}_i - \overline{b}_i^{k})}{\mathcal{O}_{2} \sum_{i=1}^{N_f}q_i\overline{b}_i^{k}}}$
\end{lemma}
\begin{proof}
$D(\textbf{b},W_1)$ can be written as 
\begin{align}
\label{optimize_eqn3_p1}
\sum_{i=1}^{N_f}q_i(\overline{b}_i - \overline{b}_i^{k})\big(\mathcal{O}_{1}W_1 - \zeta_{tot}\sum_{i=1}^{N_f}q_i(\overline{b}_i - \overline{b}_i^{k})\big)^{-1} + \sum_{i=1}^{N_f}q_i\overline{b}_i^{k}\big(\mathcal{O}_{2}W_2 - \zeta_{tot}\sum_{i=1}^{N_f}q_i\overline{b}_i^{k}\big)^{-1}, \nonumber
\end{align}
The second derivative $\frac{\partial^2 D(\textbf{b},W_1)}{\partial W_1^2}$ is hence given by 
\begin{align}
 2\mathcal{O}_{1}^2\sum_{i=1}^{N_f}q_i(\overline{b}_i - \overline{b}_i^{k})\big(\mathcal{O}_{1}W_1 - \zeta_{tot}\sum_{i=1}^{N_f}q_i(\overline{b}_i - \overline{b}_i^{k})\big)^{-3} + 2\mathcal{O}_{2}^2\sum_{i=1}^{N_f}q_i\overline{b}_i^{k}\big(\mathcal{O}_{2}W_2 - \zeta_{tot}\sum_{i=1}^{N_f}q_i\overline{b}_i^{k}\big)^{-3},  \nonumber 	
\end{align}
The stability conditions require that $\mu_1 = \mathcal{O}_{1}W_1 > \zeta_{tot}\sum_{i=1}^{N_f}q_i(\overline{b}_i - \overline{b}_i^{k})$ and $\mu_2 =\mathcal{O}_{2}W_2 > \zeta_{tot}\sum_{i=1}^{N_f}q_i\overline{b}_i^{k}$. Also, $\overline{b}_i \geq \overline{b}_i^{k}$ by definition. Hence, $\frac{\partial^2 D(\textbf{b},W_1)}{\partial W_1^2}  > 0$, and the objective function is a convex function of $W_1$. The optimal bandwidth allocation can be obtained from the \ac{KKT} conditions similar to problems \textbf{P2} and \textbf{P3}, with the details omitted for brevity.
\end{proof}
Given $W_1^*$ from the bandwidth allocation subproblem, the caching probability subproblem can be written as
\begin{align}
\textbf{P5:} \quad \quad&\underset{\textbf{b}}{\text{min}} \quad D(\textbf{b},W_1^*) \\
&\textrm{s.t.}\quad  (\ref{const110}), (\ref{const111}),  (\ref{stab1}), (\ref{stab2})  \nonumber 
\end{align}
The caching probability subproblem \textbf{P5} is a sum of two fractional functions, where the first fraction is in the form of a concave over convex functions while the second fraction is in the form of a convex over concave functions. The first fraction structure, i.e., concave over convex functions, renders solving this problem using fractional programming (FP) very challenging.\footnote{A quadratic transform technique for tackling the multiple-ratio concave-convex FP problem is recently used to solve a minimization of fractional functions that has the form of convex over concave functions, whereby an equivalent problem is solved with the objective function reformulated as a difference between convex minus concave functions \cite{shen2018fractional}.} 
\textcolor{black}{Moreover, the constraint (\ref{stab1}) is concave w.r.t. $\textbf{b}$. 
Hence, we adopt the interior point method to obtain local optimal solution of $\textbf{b}$ given the optimal bandwidth $W_1^*$, which depends on the initial value input to the algorithm \cite{boyd2004convex}. Nonetheless, we can increase the probability to find a near-optimal solution of problem \textbf{P5} by using the interior point method with multiple random initial values and then picking the solution with lowest weighted average delay. The explained procedure is repeated until the value of \textbf{P4}'s objective function converges to a pre-specified accuracy.}
\section{Numerical Results}
\begin{table}[ht]
\caption{Simulation Parameters} 
\centering 
\begin{tabular}{c c  c} 
\hline\hline 
Description & Parameter & Value  \\ [0.5ex] 
\hline 
\textcolor{black}{System bandwidth} & W & \SI{20}{\mega\hertz}  \\ 
\ac{BS} transmission power & $P_b$ & \SI{43}{\deci\bel\of{m}}  \\
\ac{D2D} transmission power & $P_d$ & \SI{23}{\deci\bel\of{m}}  \\
Displacement standard deviation & $\sigma$ & \SI{10}{\metre} \\ 
Popularity index&$\beta$&1\\
Path loss exponent&$\alpha$&4\\
Library size&$N_f$&500 files\\
Cache size per device&$M$&10 files\\
Average number of devices per cluster&$\overline{n}$&5\\
Density of $\Phi_p$&$\lambda_{p}$&20 clusters/\SI{}{km}$^2$ \\
Average content size&$\textcolor{black}{\overline{S}}$&\SI{5}{MBits} \\		
$\sir$ threshold&$\theta$&\SI{0}{\deci\bel}\\
\textcolor{black}{Total request arrival rate}&$\zeta_{tot}$&\SI{2}{request/sec}\\
\hline 
\end{tabular}
\label{ch3:table:sim-parameter} 
\end{table}
At first, we validate the developed mathematical model via Monte Carlo simulations. Then we benchmark the proposed caching scheme against conventional caching schemes. Unless otherwise stated, the network parameters are selected as shown in Table \ref{ch3:table:sim-parameter}. 
\subsection{Offloading Gain Results}
\begin{figure}[!h]
	\begin{center}
		\includegraphics[width=2.5in]{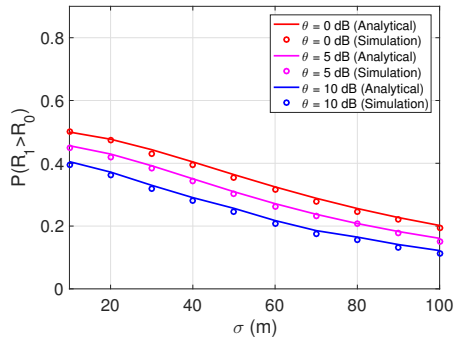}			
		\caption {The probability that the \ac{D2D} achievable rate is greater than a threshold $R_0$ versus standard deviation $\sigma$.}
		\label{prob_r_geq_r0}
	\end{center}
\vspace{-0.5cm}
\end{figure}
In this subsection, we present the offloading gain performance for the proposed caching model. 
In Fig.~\ref{prob_r_geq_r0}, we verify the accuracy of the analytical results for the probability $\mathbb{P}(R_1>R_0)$. The theoretical and simulated results are plotted together, and they are consistent. We can observe that the probability $\mathbb{P}(R_1>R_0)$ decreases monotonically with the increase of $\sigma$. This is because as $\sigma$ increases, the serving distance increases and the inter-cluster interfering distance between out-of-cluster interferers and the typical device decreases, and equivalently, the $\sir$ decreases. It is also shown that $\mathbb{P}(R_1>R_0)$ decreases with the $\sir$ threshold $\theta$ as the channel becomes more prone to be in outage when increasing the $\sir$ threshold $\theta$.
\begin{figure*}
\centering
  \subfigure[$p=p^*$. ]{\includegraphics[width=2.0in]{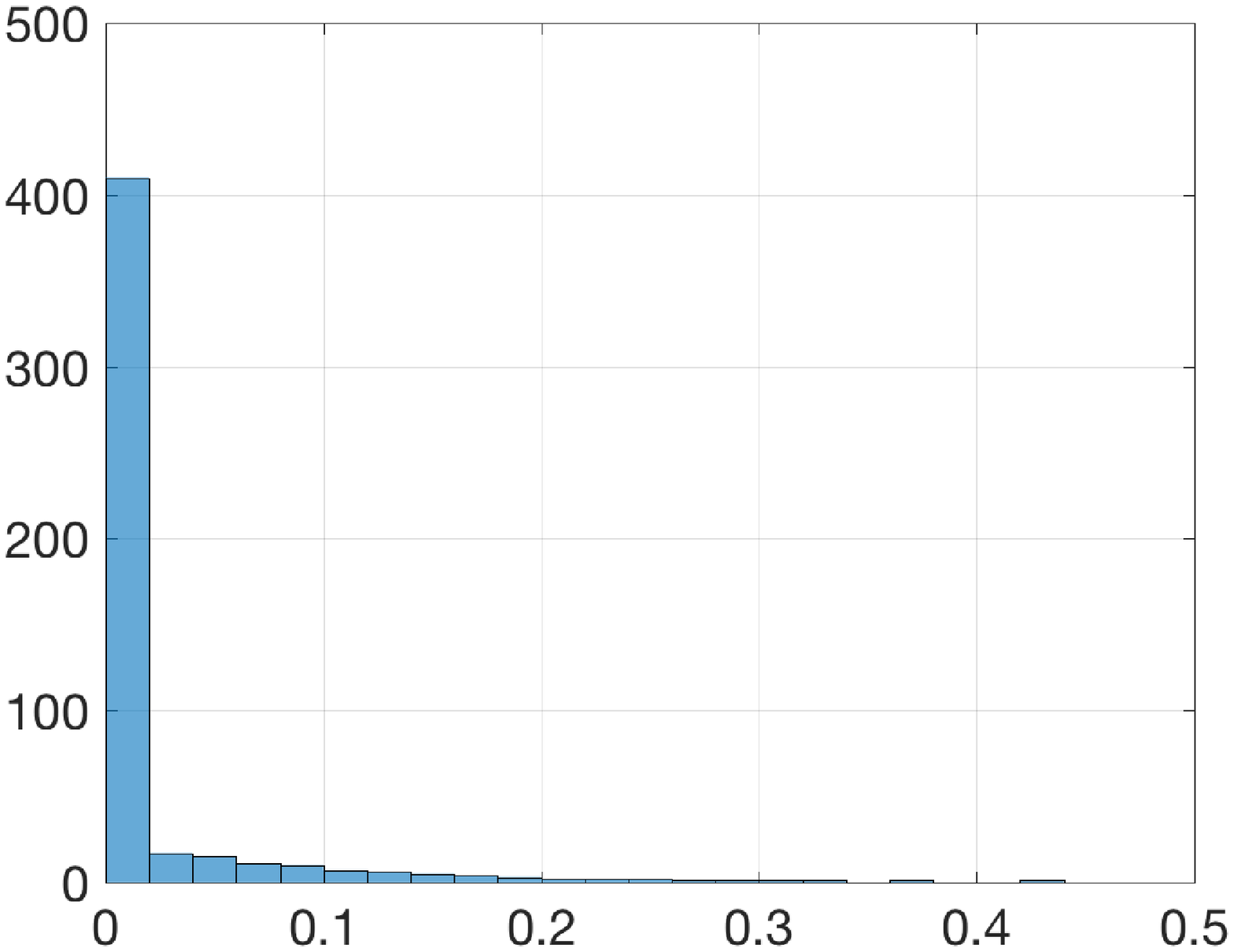}		
\label{histogram_b_i_p_star}}
\subfigure[$p > p^*$.]{\includegraphics[width=2.0in]{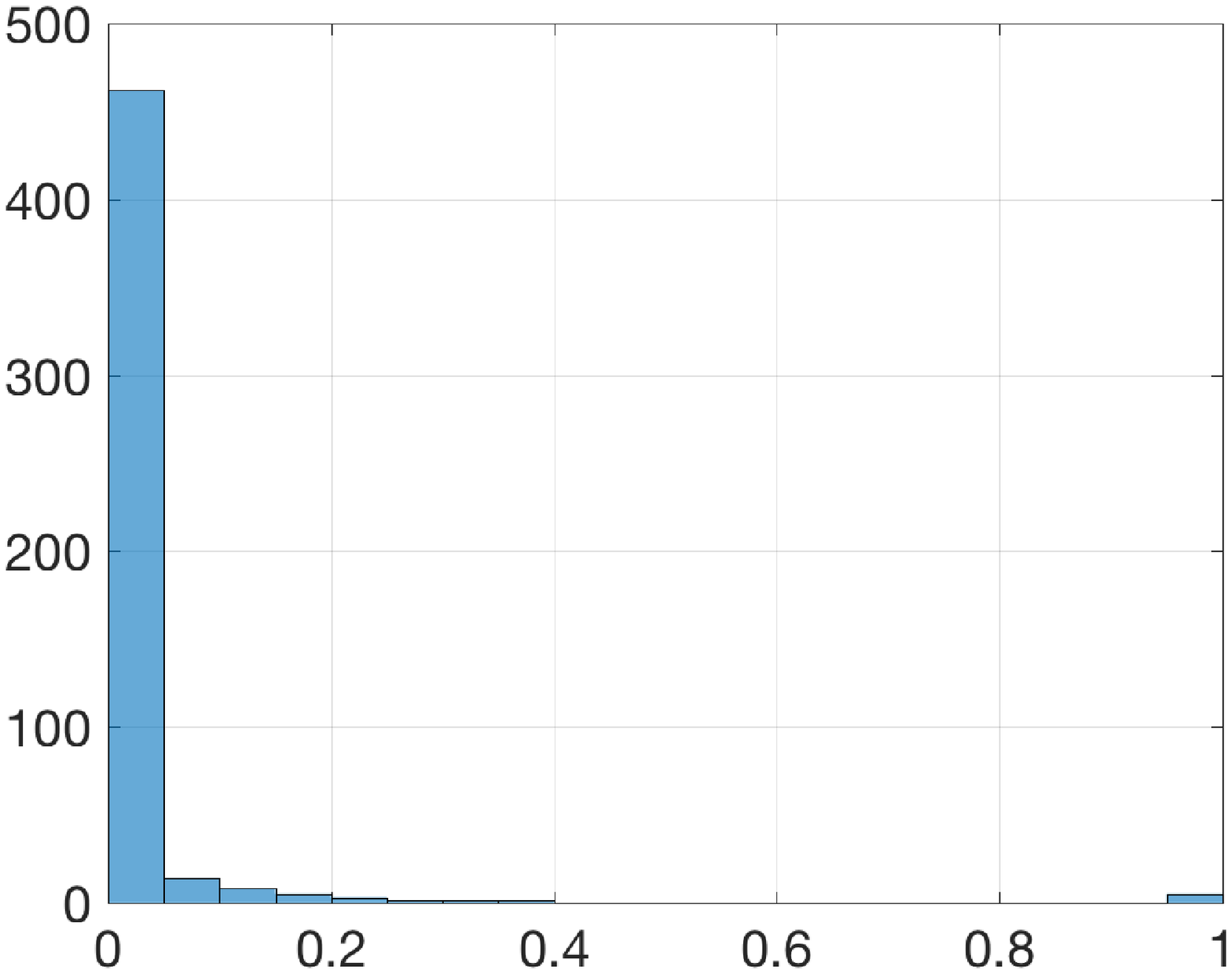}		
  \label{histogram_b_i_p_leq_p_star}}
\caption{Histogram of the optimal caching probability $\textbf{b}^*$ when (a) $p=p^*$ and (b) $p > p^*$.}
\label{histogram_b_i}
\vspace{-0.5cm}
\end{figure*}
To show the effect of $p$ on the caching probability, in Fig.~\ref{histogram_b_i}, we plot the histogram of the optimal caching probability at different values of $p$, where $p=p^*$ in Fig.~\ref{histogram_b_i_p_star} and $p>p^*$ in Fig.~\ref{histogram_b_i_p_leq_p_star}. It is clear from the histograms that the optimal caching probability $\textbf{b}^*$ tends to be more skewed when $p > p^*$, i.e., when $\mathbb{P}(R_1>R_0)$ decreases. This shows that file sharing is more difficult when $p$ is larger than the optimal access probability. More precisely, for $p>p^*$, the outage probability is high due to the aggressive interference. In such a low coverage probability regime, each device tends to cache the most popular files leading to fewer opportunities of content transfer between the devices. 
 \begin{figure}[!h]
	\begin{center}
		\includegraphics[width=2.5in]{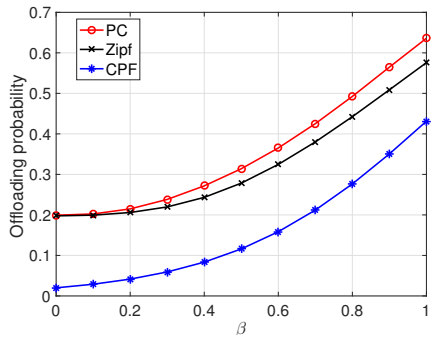}			
		\caption {The offloading probability versus the popularity of files $\beta$ under different caching schemes, namely, \ac{PC}, Zipf, and \ac{CPF}, $\overline{n}=10$, $\sigma=$\SI{5}{\metre}.}			
		\label{offloading_gain_vs_beta}
	\end{center}
\vspace{-0.5cm}
\end{figure}

Last but not least, Fig.~\ref{offloading_gain_vs_beta} manifests the prominent effect of the files' popularity on the offloading gain. We compare the offloading gain of three different caching schemes, namely, the proposed \ac{PC}, Zipf's caching (Zipf), and \ac{CPF}. We can see that the offloading gain under the \ac{PC} scheme attains the best performance as compared to other schemes. Also, we note that both \ac{PC} and Zipf schemes encompass the same offloading gain when $\beta=0$ owing to the uniformity of content popularity.  
\subsection{\textcolor{black}{Energy Consumption Results}}				
\begin{figure}[!h]
	\begin{center}
		\includegraphics[width=2.5in]{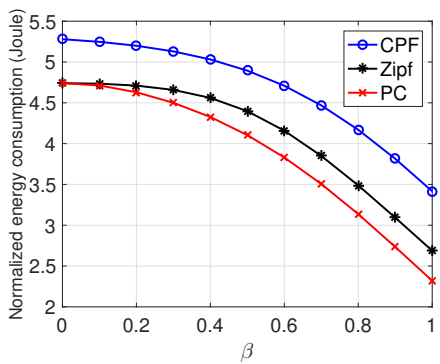}			
		\caption {Normalized energy consumption versus popularity exponent $\beta$.}			
		\label{energy_vs_beta}
	\end{center}
\vspace{-0.5cm}
\end{figure}
The results in this part are given for the energy consumption. 
Fig.~\ref{energy_vs_beta} shows the energy consumption, \textcolor{black}{normalized to the mean number of devices per cluster}, versus $\beta$ under different caching schemes, namely, \ac{PC}, Zipf, and \ac{CPF}. We can see that the minimized energy consumption under the proposed \ac{PC} scheme attains the best performance as compared to other schemes. Also, it is clear that the consumed energy decreases with $\beta$. This can be justified by the fact that as $\beta$ increases, fewer files are frequently requested which are more likely to be cached among the devices under \ac{PC}, \ac{CPF}, and the Zipf schemes. These few files therefore are downloadable from the devices via low power \ac{D2D} communication. 

\begin{figure}[!h]
	\begin{center}
		\includegraphics[width=2.5in]{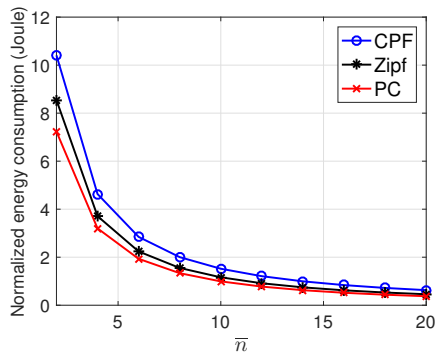}		
		\caption {\textcolor{black}{Normalized energy consumption versus the mean number of devices per cluster.}} 
		\label{energy_vs_n}
	\end{center}
\vspace{-0.5cm}
\end{figure}
We plot the normalized energy consumption versus the \textcolor{black}{mean number of devices per cluster} in Fig.~\ref{energy_vs_n}. First, we see that energy consumption decreases with the mean number of devices per cluster. As the number of devices per cluster increases, it is more probable to obtain requested files via low power \ac{D2D} communication. When the number of devices per cluster is relatively large, the normalized energy consumption tends to flatten as most of the content becomes cached at the cluster devices. 
\subsection{Delay Results}						
\begin{figure*}  [!h]
\centering
  \subfigure[Weighted average delay versus the popularity exponent $\beta$. ]{\includegraphics[width=2.0in]{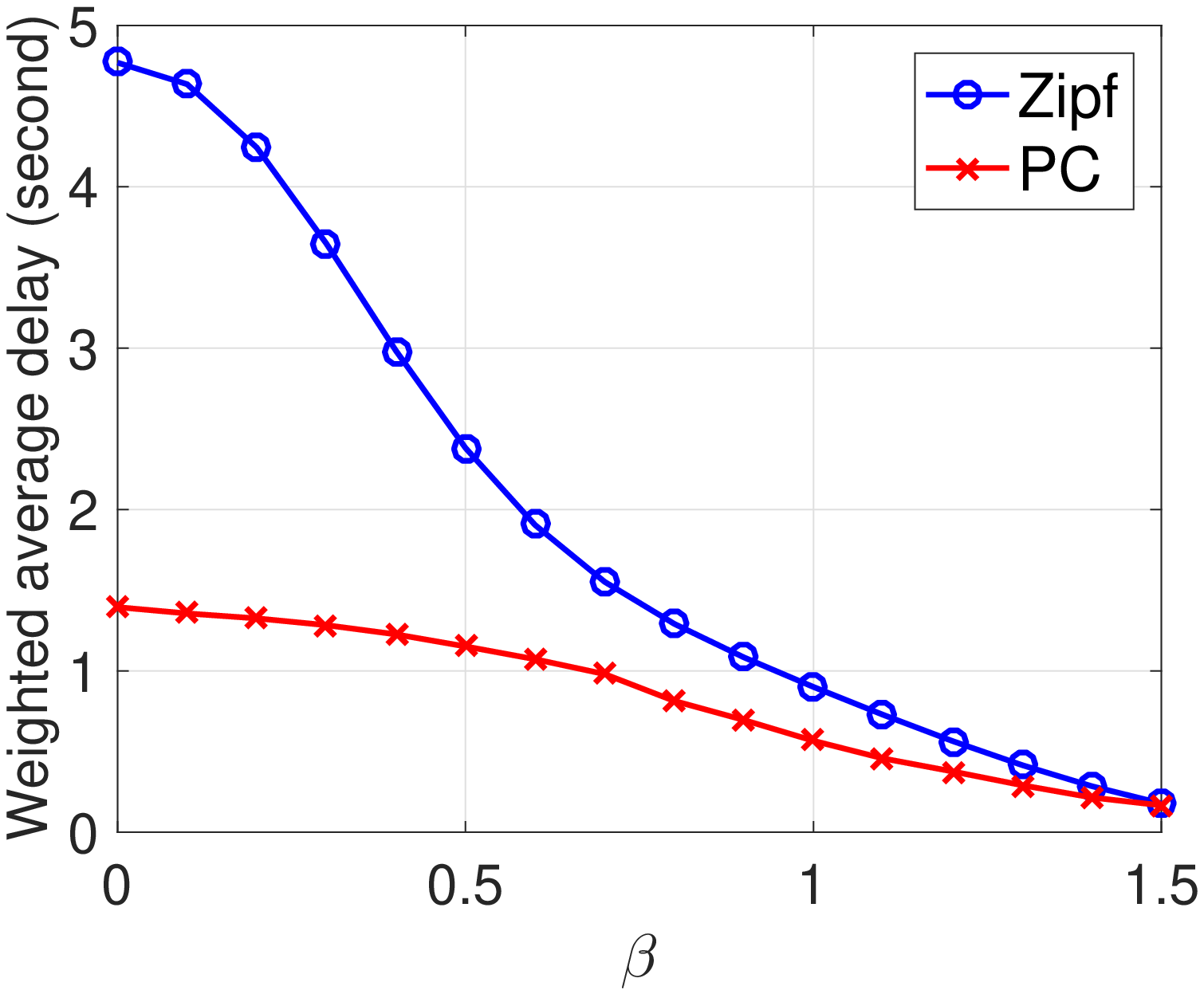}			
\label{delay_compare}}
\subfigure[Optimal allocated bandwidth versus the popularity exponent $\beta$.]{\includegraphics[width=2.0in]{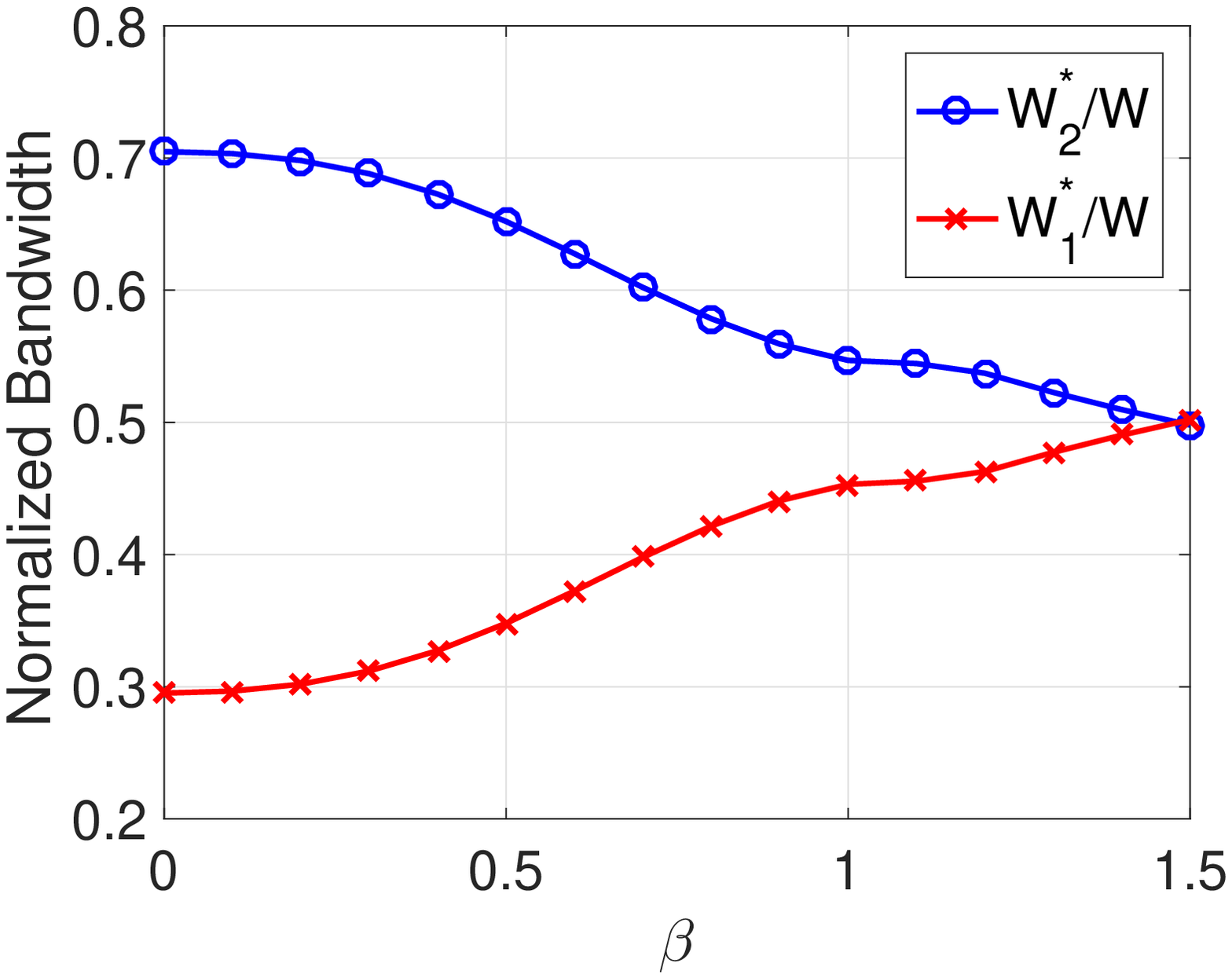}		
  \label{BW_compare}}
\caption{\textcolor{black}{Evaluation and comparison of average delay for the proposed joint \ac{PC} and bandwidth allocation scheme with the Zipf's baseline scheme against popularity exponent $\beta$, $N_f = 100$, $M = 4$, $k = 8$.}}
\label{delay-analysis}
\vspace{-0.5cm}
\end{figure*}				

\textcolor{black}{The results in this part are devoted to the average delay metric. The performance of the proposed joint \ac{PC} and bandwidth allocation scheme is evaluated in Fig.~\ref{delay-analysis}, and the optimized bandwidth allocation is also shown. Firstly, in Fig.~\ref{delay_compare}, we compare the average delay for two different caching schemes, namely, \ac{PC}, and Zipf's scheme. We can see that the minimized average delay under the proposed joint \ac{PC} and bandwidth allocation scheme attains substantially better performance as compared to the Zipf's scheme with fixed bandwidth allocation (i.e., $W_1=W_2=W/2$). Also, we see that, in general, the average delay monotonically decreases with $\beta$ when a fewer number of files undergoes the highest demand. 
Secondly, Fig.~\ref{BW_compare} manifests the effect of the files' popularity $\beta$ on the allocated bandwidth. It is shown that optimal \ac{D2D} allocated bandwidth $W_1^*$ continues increasing with $\beta$. This can be interpreted as follows. When $\beta$ increases, a fewer number of files become highly demanded. These files can be entirely cached among the devices. To cope with such a larger number of requests served via the \ac{D2D} communication, the \ac{D2D} allocated bandwidth needs to be increased.}

\textcolor{black}{Fig.~\ref{scaling} shows the geometrical scaling effects on the system performance, e.g., the effect of clusters' density $\lambda_p$ and the displacement standard deviation $\sigma$ on the \ac{D2D} coverage probability $\pcd$, optimal allocated bandwidth $W_1^*$, and the average delay. In Fig.~\ref{cache_size}, we plot the \ac{D2D} coverage probability $\pcd$ versus the displacement standard deviation $\sigma$ for different clusters' density $\lambda_p$. It is clear from the plot that $\pcd$ monotonically decreases with both $\sigma$ and $\lambda_p$. Obviously, increasing $\sigma$ and $\lambda_p$ results in larger serving distance, i.e., higher path-loss effect, and shorter interfering distance, i.e., higher interference power received by the typical device, respectively. This explains the encountered degradation for $\pcd$ with $\sigma$ and $\lambda_p$. }
In Fig.~\ref{optimal-w}, we plot the optimal allocated bandwidth $W_1^*$ normalized to $W$ versus the displacement standard deviation $\sigma$ for different clusters' density $\lambda_p$. In this case too it is quite obvious that $W_1^*$ tends to increase with both $\sigma$ and $\lambda_p$. This behavior can be directly understood from (\ref{optimal-w-1}) where $W_1^*$ is inversely proportional to $\mathcal{O}_{1} = \frac{\pcd {\rm log}_2(1+\theta)}{\overline{S}}$, and $\pcd$ decreases with $\sigma$ and $\lambda_p$ as discussed above. More precisely, while the \ac{D2D} service rate $\mu_1$ tends to decrease with the decrease of $\pcd$ since $\mu_1 = \frac{\pcd W_1{\rm log}_2(1+\theta)}{\overline{S}}$, the optimal allocated bandwidth $W_1^*$ tends to increase with the decrease of $\pcd$ to compensate for the service rate degradation, and eventually, minimizing the weighted average delay.
In Fig.~\ref{av-delay}, we plot the weighted average delay versus the displacement standard deviation $\sigma$ for different clusters' density $\lambda_p$. Following the same interpretations as in Fig.~\ref{cache_size} and Fig.~\ref{optimal-w}, we can notice that the weighted average delay monotonically increases with  $\sigma$ and $\lambda_p$ due to the decrease of the \ac{D2D} coverage probability $\pcd$ and the \ac{D2D} service rate $\mu_1$ with $\sigma$ and $\lambda_p$.
\begin{figure*} [t!]	
\vspace{-0.5cm}
\centering
  \subfigure[\ac{D2D} coverage probability $\pcd$ versus the displacement standard deviation $\sigma$ for different clusters' density $\lambda_p$.]{\includegraphics[width=2.0in]{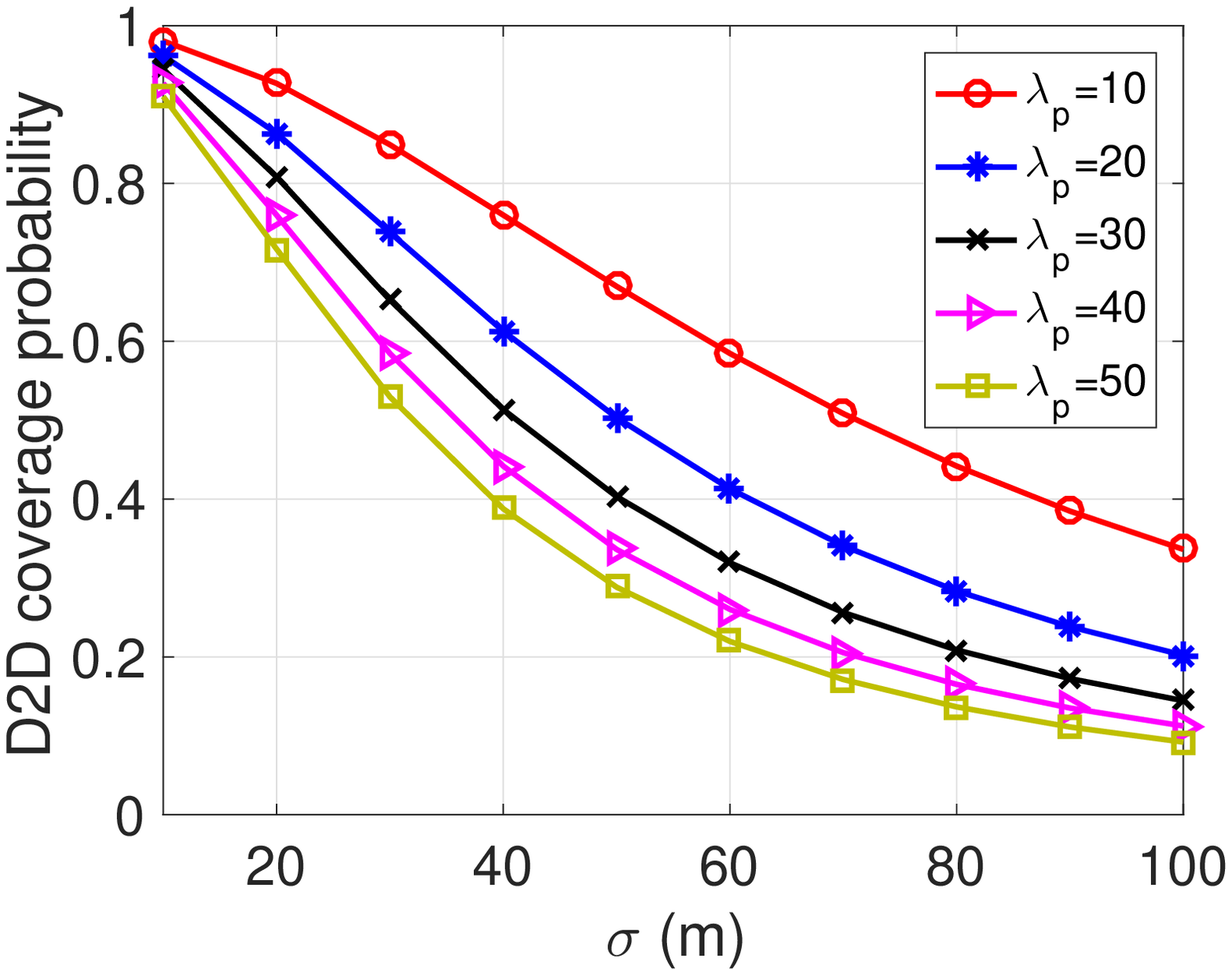}
\label{cache_size}}
\subfigure[Optimal allocated bandwidth versus the displacement standard deviation $\sigma$ for different clusters' density $\lambda_p$.]{\includegraphics[width=2.0in]{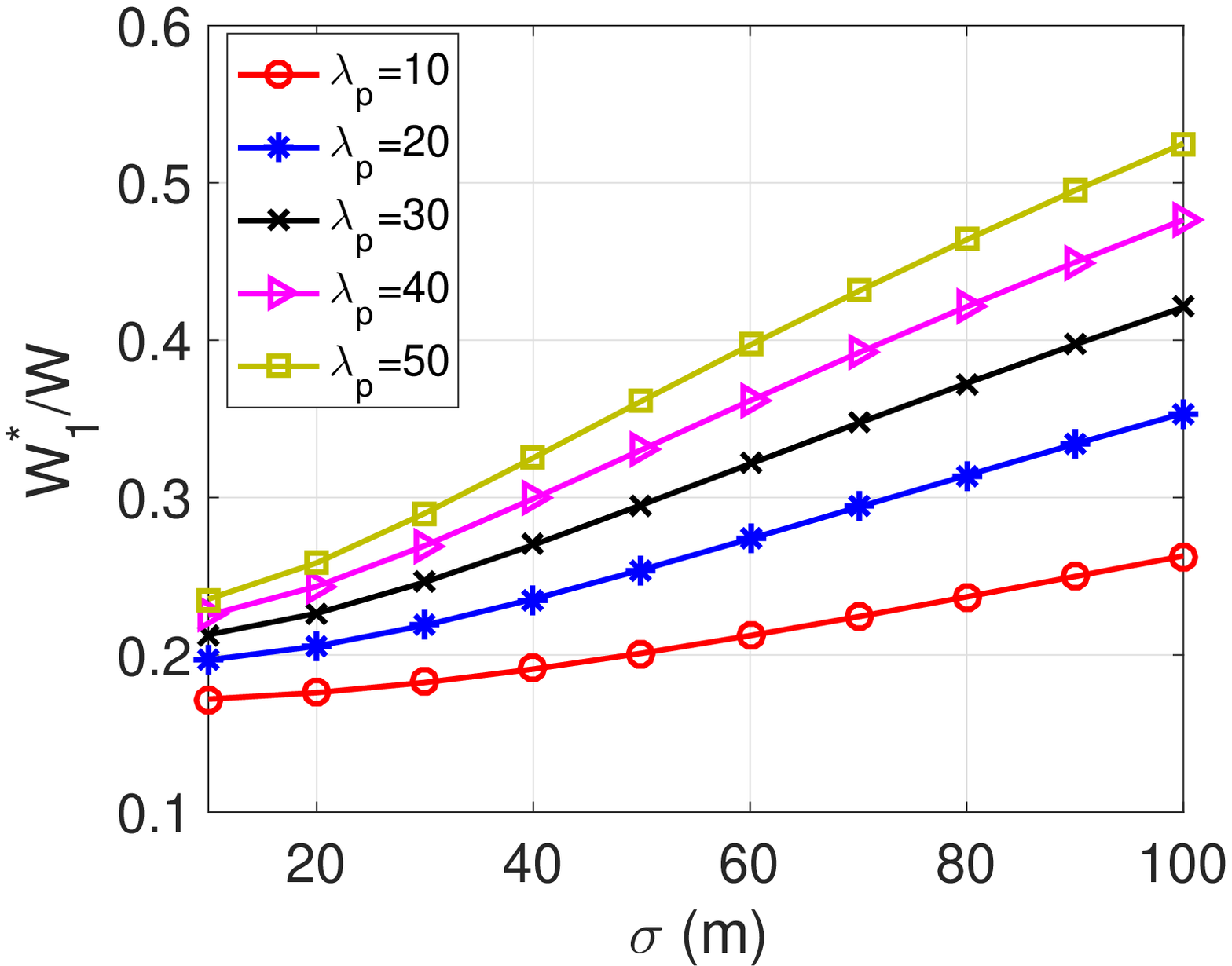}		 
  \label{optimal-w}}
\subfigure[Weighted average delay versus the displacement standard deviation $\sigma$ for different clusters' density $\lambda_p$.]{\includegraphics[width=2.0in]{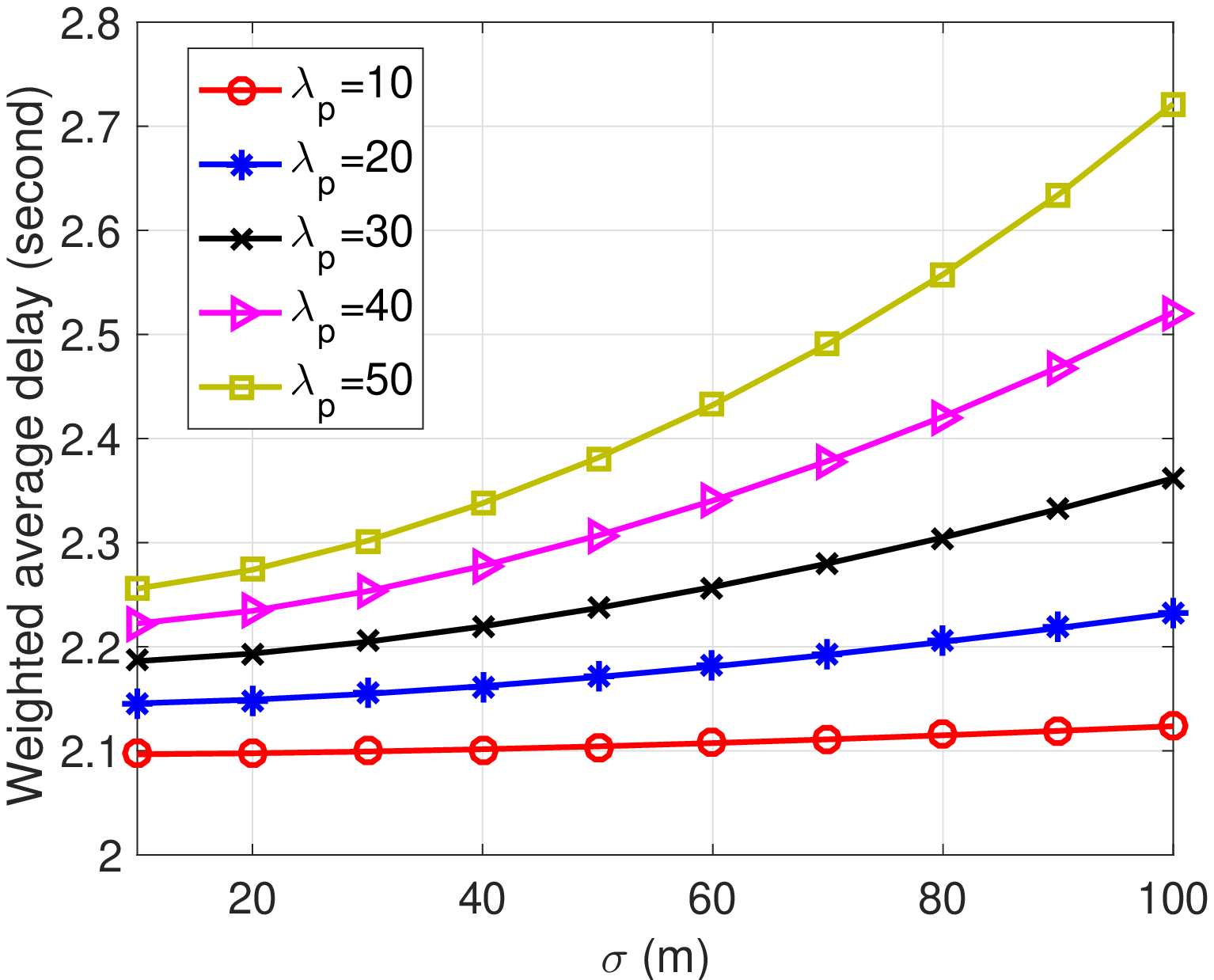}		 
  \label{av-delay}}
\caption{Effect of geometrical parameters, e.g., clusters' density $\lambda_p$ and the displacement standard deviation $\sigma$ on the system performance, $\beta= 0.5$, $N_f = 100$, $M = 4$, $k = 8$.}	
\vspace{-0.6cm}
\label{scaling}
\end{figure*}
\section{Conclusion}
In this work, we conduct a comprehensive analysis of the joint communication and caching for a clustered \ac{D2D} network with random probabilistic caching incorporated at the devices. We first maximize the offloading gain of the proposed system by jointly optimizing the channel access and caching probability. We obtain the optimal channel access probability, and the optimal caching probability is then characterized. We show that deviating from the optimal access probability $p^*$ makes file sharing more difficult. More precisely, the system is too conservative for small access probabilities, while the interference is too aggressive for larger access probabilities.  Then, we minimize the energy consumption of the proposed clustered \ac{D2D} network. We formulate the energy minimization problem and show that it is convex and the optimal caching probability is obtained. We show that a content with a large size or low popularity has a small probability to be cached. Finally, we adopt a queuing model for the devices' traffic within each cluster to investigate the network average delay. Two M/G/1 queues are employed to model the \ac{D2D} and \ac{BS}-to-Device communications. We then derive an expression for the weighted average delay per request. We observe that the average delay is dependent on the caching probability and bandwidth allocated, which control respectively the arrival rates and service rates for the two modeling queues. Therefore, we minimize the per request weighted average delay by jointly optimizing bandwidth allocation between \ac{D2D} and \ac{BS}-to-Device communication and the caching probability. The delay minimization problem is shown to be non-convex. Applying the \ac{BCD} optimization technique, the joint minimization problem can be solved in an iterative manner. Results show up to $10\%$, $17\%$, and $300\%$ improvement gain in the offloading gain, energy consumption, and average delay, respectively, compared to the Zipf's caching technique.
\begin{appendices}

\section{Proof of lemma 1}
Laplace transform of the inter-cluster aggregate interference $I_{\Phi_p^{!}}$ can be evaluated as
\begin{align}
 \mathscr{L}_{I_{\Phi_p^{!}}}(s) &= \mathbb{E} \Bigg[e^{-s \sum_{\Phi_p^{!}} \sum_{y \in \mathcal{B}^p}  g_{y_{x}}  \lVert x + y\rVert^{-\alpha}} \Bigg] \nonumber \\ 
   &= \mathbb{E}_{\Phi_p} \Bigg[\prod_{\Phi_p^{!}} \mathbb{E}_{\Phi_{cp},g_{y_{x}}} \prod_{y \in \mathcal{B}^p}  e^{-s g_{y_{x}}   \lVert x + y\rVert^{-\alpha}} \Bigg]  \nonumber \\
  &= \mathbb{E}_{\Phi_p} \Bigg[\prod_{\Phi_p^{!}} \mathbb{E}_{\Phi_{cp}} \prod_{y \in \mathcal{B}^p} \mathbb{E}_{g_{y_{x}}}  e^{-s g_{y_{x}}   \lVert x + y\rVert^{-\alpha}} \Bigg]  \nonumber \\
    &\overset{(a)}{=} \mathbb{E}_{\Phi_p} \Bigg[\prod_{\Phi_p^{!}} \mathbb{E}_{\Phi_{cp}} \prod_{y \in \mathcal{B}^p} \frac{1}{1+s \lVert x + y\rVert^{-\alpha}} \Bigg]  \nonumber \\
     &\overset{(b)}{=} \mathbb{E}_{\Phi_p} \prod_{\Phi_p^{!}} {\rm exp}\Big(-p\overline{n} \int_{\mathbb{R}^2}\Big(1 - \frac{1}{1+s \lVert x + y\rVert^{-\alpha}}\Big)f_Y(y)\dd{y}\Big) 	 \nonumber	
     \\
          &\overset{(c)}{=}  {\rm exp}\Bigg(-\lambda_p \int_{\mathbb{R}^2}\Big(1 -  {\rm exp}\Big(-p\overline{n} \int_{\mathbb{R}^2}\Big(1 - \frac{1}{1+s \lVert x + y\rVert^{-\alpha}}\Big)f_Y(y)\dd{y}\Big)\dd{x}\Bigg)	             	
\end{align}
where $\varphi(s,v) = \int_{u=0}^{\infty}\frac{s}{s+ u^{\alpha}}f_U(u|v)\dd{u}$; (a) follows from the Rayleigh fading assumption, (b) 
follows from the \ac{PGFL} of Gaussian \ac{PPP} $\Phi_{cp}$, and (c) follows from the \ac{PGFL} of the parent \ac{PPP} $\Phi_p$. By using change of variables $z = x + y$ with $\dd z = \dd y$, we proceed as
\begin{align}
\mathscr{L}_{I_{\Phi_p^{!}}}(s) &\overset{}{=}  {\rm exp}\Bigg(-\lambda_p \int_{\mathbb{R}^2}\Big(1 -  {\rm exp}\Big(-p\overline{n} \int_{\mathbb{R}^2}\Big(1 - \frac{1}{1+s \lVert z\rVert^{-\alpha}}\Big)f_Y(z-x)\dd{y}\Big)\dd{x}\Bigg)			 \nonumber	
          \\
          &\overset{(d)}{=}  {\rm exp}\Bigg(-2\pi\lambda_p \int_{v=0}^{\infty}\Big(1 -  {\rm exp}\Big(-p\overline{n} \int_{u=0}^{\infty}\Big(1 - \frac{1}{1+s  u^{-\alpha}}\Big)f_U(u|v)\dd{u}\Big)v\dd{v}\Bigg)		 \nonumber
          \\
         &=  {\rm exp}\Bigg(-2\pi\lambda_p \int_{v=0}^{\infty}\Big(1 -  {\rm exp}\big(-p\overline{n} \int_{u=0}^{\infty}
         \frac{s}{s+ u^{\alpha}}f_U(u|v)\dd{u}\big)v\dd{v}\Big)\Bigg)		 \nonumber
         \\
         &=  {\rm exp}\Big(-2\pi\lambda_p \int_{v=0}^{\infty}\Big(1 -  {\rm e}^{-p\overline{n} \varphi(s,v)}\Big)v\dd{v}\Big),	
\end{align}
where (d) follows from converting the cartesian coordinates to the polar coordinates with $u=\lVert z\rVert$. 
\textcolor{black}{To clarify how in (d) the normal distribution $f_Y(z-x)$ is converted to the Rice distribution $f_U(u|v)$, consider a remote cluster centered at $x \in \Phi_p^!$, with a distance $v=\lVert x\rVert$ from the origin. Every interfering device belonging to the cluster centered at $x$ has its coordinates in $\mathbb{R}^2$ chosen independently from a Gaussian distribution with standard deviation $\sigma$. Then, by definition, the distance from such an interfering device to the origin, denoted as $u$, has a Rice distribution, denoted as $f_U(u|v)=\frac{u}{\sigma^2}\mathrm{exp}\big(-\frac{u^2 + v^2}{2\sigma^2}\big) I_0\big(\frac{uv}{\sigma^2}\big)$, where $I_0$ is the modified Bessel function of the first kind with order zero and $\sigma$ is the scale parameter.
Hence, Lemma 1 is proven.}

\section {Proof of lemma 2}
Laplace transform of the intra-cluster aggregate interference $I_{\Phi_c}$, conditioning on the distance $v_0$ from the cluster center to the origin, see Fig~\ref{distance}, is written as
\begin{align}
 \mathscr{L}_{I_{\Phi_c} }(s|v_0) &= \mathbb{E} \Bigg[e^{-s \sum_{y \in \mathcal{A}^p}  g_{y_{x_0}}  \lVert x_0 + y\rVert^{-\alpha}} \Bigg] \nonumber 
 \\ 
   &=  \mathbb{E}_{\Phi_{cp},g_{y_{x_0}}} \prod_{y \in\mathcal{A}^p}  e^{-s g_{y_{x_0}}   \lVert x_0 + y\rVert^{-\alpha}}  \nonumber 
   \\
    &=  \mathbb{E}_{\Phi_{cp}} \prod_{y \in\mathcal{A}^p}  \mathbb{E}_{g_{y_{x_0}}} e^{-s g_{y_{x_0}}   \lVert x_0 + y\rVert^{-\alpha}}  \nonumber 
   \\
    &\overset{(a)}{=} \mathbb{E}_{\Phi_{cp}} \prod_{y \in\mathcal{A}^p} \frac{1}{1+s \lVert x_0 + y\rVert^{-\alpha}} 
     \nonumber 
     \\
     &\overset{(b)}{=}  {\rm exp}\Big(-p\overline{n} \int_{\mathbb{R}^2}\Big(1 - \frac{1}{1+s \lVert x_0 + y\rVert^{-\alpha}}\Big)f_{Y}(y)\dd{y}\Big) 	 \nonumber	
     \\
        &\overset{(c)}{=} {\rm exp}\Big(-p\overline{n} \int_{\mathbb{R}^2}\Big(1 - \frac{1}{1+s \lVert z_0\rVert^{-\alpha}}\Big)f_{Y}(z_0-x_0)\dd{z_0}\Big)
\end{align}
where (a) follows from the Rayleigh fading assumption, (b)  follows from the \ac{PGFL} of the Gaussian \ac{PPP} $\Phi_{cp}$, (c) follows from changing of variables $z_0 = x_0 + y$ with $\dd z_0 = \dd y$. By converting the cartesian coordinates to the polar coordinates, with $h=\lVert z_0\rVert$, we get  
\begin{align}
 \mathscr{L}_{I_{\Phi_c} }(s|v_0) &\overset{}{=}  {\rm exp}\Big(-p\overline{n} \int_{h=0}^{\infty}\Big(1 - \frac{1}{1+s  h^{-\alpha}}\Big)f_H(h|v_0)\dd{h}\Big)		 \nonumber
          \\
         &=  {\rm exp}\Big(-p\overline{n} \int_{h=0}^{\infty}\frac{s}{s+ h^{\alpha}}f_H(h|v_0)\dd{h}\Big)	
 \end{align}
 By neglecting the correlation of the intra-cluster interfering distances as in \cite{clustered_twc}, i.e., the common part $x_0$ in the intra-cluster interfering distances $\lVert x_0 + y\rVert$, $y \in \mathcal{A}^p$, we get.
\begin{align}
         \mathscr{L}_{I_{\Phi_c} }(s) &\approx  {\rm exp}\Big(-p\overline{n} \int_{h=0}^{\infty}\frac{s}{s+ h^{\alpha}}f_H(h)\dd{h}\Big)
\end{align}
Similar to the serving distance \ac{PDF} $f_R(r)$, since both the typical device and a potential interfering device have their locations drawn from a normal distribution with variance $\sigma^2$ around the cluster center, then by definition, the intra-cluster interfering distance has a Rayleigh distribution with parameter $\sqrt{2}\sigma$, and given by $f_H(h)=  \frac{h}{2 \sigma^2} {\rm e}^{\frac{-h^2}{4 \sigma^2}}$. Hence, Lemma 2 is proven. 
\section {Proof of lemma 3}
First, to prove concavity, we proceed as follows. 
\begin{align}
\frac{\partial \mathbb{P}_o}{\partial b_i} &= q_i + q_i\big(\overline{n}(1-b_i)e^{-\overline{n}b_i}-(1-e^{-\overline{n}b_i})\big)\mathbb{P}(R_1>R_0)		 \\
\frac{\partial^2 \mathbb{P}_o}{\partial b_i^2} &= -q_i\big(\overline{n}e^{-\overline{n}b_i} + \overline{n}^2(1-b_i)e^{-\overline{n}b_i} + \overline{n}e^{-\overline{n}b_i}\big)\mathbb{P}(R_1>R_0)
\end{align}
It is clear that the second derivative $\frac{\partial^2 \mathbb{P}_o}{\partial b_i^2}$ is always negative, and 
$\frac{\partial^2 \mathbb{P}_o}{\partial b_i \partial b_j}=0$ for all $i\neq j$. Hence, the Hessian matrix \textbf{H}$_{i,j}$ of $\mathbb{P}_o(p^*,\textbf{b})$ w.r.t. $\textbf{b}$ is negative semidefinite, and $\mathbb{P}_o(p^*,\textbf{b})$ is a concave function of $\textbf{b}$. Also, the constraints are linear, which imply that the necessity and sufficiency conditions for optimality exist. The dual Lagrangian function and the \ac{KKT} conditions are then employed to solve \textbf{P2}.
The \ac{KKT} Lagrangian function of the energy minimization problem is given by
\begin{align}
\mathcal{L}(\textbf{b},w_i,\mu_i,v) =& \sum_{i=1}^{N_f} q_i b_i + q_i(1- b_i)(1-e^{-b_i\overline{n}})\mathbb{P}(R_1>R_0) \nonumber \\
&+ v(M-\sum_{i=1}^{N_f} b_i) + \sum_{i=1}^{N_f} w_i (b_i-1) - \sum_{i=1}^{N_f} \mu_i b_i
\end{align}
where $v, w_i, \mu_i$ are the dual equality and two inequality constraints, respectively. Now, the optimality conditions are written as		
\begin{align}
\label{grad}
\grad_{\textbf{b}} \mathcal{L}(\textbf{b}^*,w_i^*,\mu_i^*,v^*) = q_i + q_i&\big(\overline{n}(1-b_i^*)e^{-\overline{n}b_i^*}-(1-e^{-\overline{n}b_i^*})\big)\mathbb{P}(R_1>R_0) -v^* + w_i^* -\mu_i^*= 0 \\
 &w_i^* \geq 0 \\
 &\mu_i^* \leq 0 \\
 &w_i^* (b_i^* - 1) =0 \\
&\mu_i^* b_i^* = 0\\
 &(M-\sum_{i=1}^{N_f} b_i^*) = 0 
\end{align}

\begin{enumerate}
  \item $w_i^* > 0$: We have $b_i^* = 1$, $\mu_i^*=0$, and
  \begin{align}
&q_i -q_i(1-e^{-\overline{n}})\mathbb{P}(R_1>R_0)= v^* - w_i^* \nonumber \\
\label{cond1_offload}
&v^{*} <   q_i -q_i(1-e^{-\overline{n}})\mathbb{P}(R_1>R_0)
 \end{align}
  \item $\mu_i^* < 0$: We have $b_i^* = 0$, and $w_i^*=0$, and
  \begin{align}
  & q_i + \overline{n}q_i\mathbb{P}(R_1>R_0) = v^* + \mu_i^*  \nonumber \\  
  \label{cond2_offload}
  &v^{*} > q_i + \overline{n}q_i\mathbb{P}(R_1>R_0)
  \end{align}
     \item $0 <b_i^*<1$: We have $w_i^*=\mu_i^*=0$, and
  \begin{align}
\label{psii_offload}
   v^{*} =   q_i  +  q_i\big(\overline{n}(1-b_i^*)e^{-\overline{n}b_i^*} - (1 - e^{-\overline{n}b_i^*})\big)\mathbb{P}(R_1>R_0)
\end{align}
\end{enumerate}
By combining (\ref{cond1_offload}), (\ref{cond2_offload}), and (\ref{psii_offload}), with the fact that $\sum_{i=1}^{N_f} b_i^*=M$, Lemma 3 is proven.
\end{appendices}
\bibliographystyle{IEEEtran}

\bibliography{bibliography}
\end{document}